\documentclass[a4paper,11pt,leqno]{article}

\usepackage{graphicx}
\usepackage{amsmath}
\usepackage{amssymb}
\usepackage{amsthm}
\usepackage{fullpage}
\usepackage{bm}
\usepackage{xcolor}
\usepackage{url}
\usepackage{hyperref}
\usepackage{setspace}
\usepackage{chngcntr}
\usepackage[sort,round]{natbib}

\counterwithin{figure}{section}
\counterwithin{table}{section}
\numberwithin{equation}{section}

\hypersetup{ 
  pdfnewwindow=true, 
  colorlinks=true, 
  linkcolor=red,
  citecolor=brown, 
  filecolor=magenta, 
  urlcolor=cyan}

\newtheorem{proposition}{Proposition}
\newtheorem{corollary}{Corollary}
\newtheorem{theorem}{Theorem}
\numberwithin{theorem}{section}
\numberwithin{corollary}{section}
\numberwithin{proposition}{section}

\begin{document}
\title{On Time Scaling of Semivariance in a Jump-Diffusion
  Process}
\author{Rodrigue Oeuvray\footnote{Pictet Asset Management SA, route
    des Acacias 60, CH-1211 Geneva 73,
    Switzerland (\href{mailto:rodrigue.oeuvray@gmail.com}{rodrigue.oeuvray@gmail.com})}
  \and Pascal Junod\footnote{University of Applied Sciences and Arts Western
    Switzerland / HEIG-VD, route de Cheseaux 1, CH-1401
    Yverdon-les-Bains, Switzerland
    (\href{mailto:pascal.junod@heig-vd.ch}{pascal.junod@heig-vd.ch})}}
\maketitle
\thispagestyle{empty}

\begin{abstract}
  The aim of this paper is to examine the time scaling of the
  semivariance when returns are modeled by various types of
  jump-diffusion processes, including stochastic volatility
  models with jumps in returns and in volatility. In particular, 
  we derive an exact formula for the semivariance when the volatility
  is kept constant, explaining how it should be scaled when
  considering a lower frequency. We also provide and justify the use
  of a generalization of the Ball-Torous approximation of a
  jump-diffusion process, this new model appearing to deliver a more
  accurate estimation of the downside risk. We use Markov Chain
  Monte Carlo (MCMC) methods to fit our stochastic volatility
  model. For the tests, we apply our methodology to a highly
  skewed set of returns based on the Barclays US High Yield Index,
  where we compare different time scalings for the semivariance. Our
  work shows that the square root of the time horizon seems to be a
  poor approximation in the context of semivariance and that our
  methodology based on jump-diffusion processes gives much better
  results.
\end{abstract}

\section{Introduction}
\label{sec:intro}
Modern portfolio theory has shown that investing in certain asset
classes promising higher returns has always been linked with a higher
variability (also called volatility) of those returns, hence resulting
in increased risks for the investor. Hence, it is one of the
main tasks of financial engineering to accurately estimate the
variability of the return of a given asset (or portfolio of
assets) and take this figure into account in various tasks, including
risk management, finding optimal portfolio strategies for a given
risk aversion level, or derivatives pricing. 

For instance, properly estimating volatilities is essential
for applying the~\cite{BlaSch73} option pricing method;
another prominent example is the estimation of return distribution
quantiles for computing value-at-risk (VaR) figures, a method which is
typically recommended by international banking regulation authorities,
such as the Basel Committee on Banking Supervision. This approach is
also widely applied in practice in internal risk management systems of
many financial institutions worldwide. 

In practice, sufficient statistical information about the past
behavior of an asset's return is often not available, hence one
is commonly forced to use the so-called \emph{square-root-of-time}
rule. Essentially, this rule transforms high-frequency risk estimates
(for instance, gathered with a 1-day period over several years) into a
lower frequency $T$ (like a 1-year period) by multiplying the
volatility by a factor of $\sqrt{T}$. As an illustration, the Basel
regulations recommend to compute the VaR for the ten day regulatory
requirement by estimating a 1-day VaR and by multiplying this value by
$\sqrt{10}$, where the VaR is the value that solves the equation 
\begin{equation*}
  \varepsilon = \int_{-\infty}^{-\text{VaR}} \hat{f}(r)dr
\end{equation*}
given the density $\hat{f}(r)$ of the
bank's return estimated probability distribution and a confidence
level $\varepsilon$, fixed for instance to $1\%$. 

It is well-known that scaling volatilities with the
square-root-of-time is only accurate under a certain number of
assumptions that are typically not observed in practice: according to
~\cite{DanZie06},
returns need to be homoscedastic and conditionally serially
uncorrelated at all leads, an assumption slightly weaker than the one
of independently and identically distributed (iid)
returns. Dan{\'\i}elsson and Zigrand show furthermore that
for the square-root-of-time rule to be correct for all quantiles and
horizons implies the iid property of the zero-means returns, but also
that the returns are normally distributed. 

In this paper, we are interested in studying the effects of applying
the square-root-of-time rule on the \emph{semivariance} of a
continuous jump-diffusion process. As a first step, the semivariance
being a downside risk measure, we quickly recall its history and
properties in the following section. 

\subsection{Downside Risk and Semivariance}
\label{subsec:downsiderisk}
Downside risk measures have appeared in the context of portfolio
theory in the 1950's, with the development by~\cite{Mar52}
and~\cite{Roy52} of decision-making tools helping to manage
risky investment portfolios. \cite{Mar52} showed how to
exploit the averages, variances and covariances of the return
distributions of assets contained in a portfolio in order to compute
an efficient frontier on which every portfolio either maximizes the
expected return for a given variance (i.e., risk level), or minimizes
the variance for a given expected return. In the scenario of
Markowitz, a utility function, defining the investor's sensitivity to
changing wealth and risk, is used to pick the proper portfolio on the
optimal border. 

On his side, \cite{Roy52} was willing to derive a
practical method allowing to determine the best risk-return trade-off;
as he was not convinced that it is feasible to model in practice
the sensitivity to risk of a human being with a utility function, he
chose to assume that an investor would prefer the investment with
the smallest probability of going below a disaster level, or a target
return. Recognizing the wisdom of this claim, \cite{Mar59}
figured out two very important points, namely that only
the downside risk is relevant for an investor, and that return
distributions might be skewed, i.e., not symmetrically distributed, in
practice. In that spirit, Markowitz suggested to use the following
variability measure, that he called a \emph{semivariance}, as it only
takes into account a subset of the return distribution: 
\begin{equation}
  \label{eq:TSV}
  \int_{-\infty}^{\tau} (\tau-r)^2f(r)dr
\end{equation}
where $f(r)$ denotes the density of the 
returns probability distribution, $R$ denotes a random variable
distributed according to $f(r)$ and $\tau$ is a return target level. If
$\tau$ is equal to $\mu_R = \int rf(r)dr$, then~\eqref{eq:TSV}
is called the \emph{below-mean} semivariance
of $R$, while if $\tau$ is arbitrary, \eqref{eq:TSV} is
called the \emph{below-target} semivariance of $R$, where $\tau$ is
defined to be the target return. In other words, only the deviations
to the left of the returns distribution average, or a fixed return
target are accounted for in the computations of the
variability. Similarly, the square root of a semivariance is called a
\emph{semideviation}, with analogy to the standard deviation. Note
that for a symmetrical, i.e., non-skewed return distribution, the
variance of a random variable $R$ is equal to twice its below-mean
semivariance. 

The~\cite{Sha66} ratio is a measure of the risk-adjusted return
of an asset, a portfolio or an investment strategy, that quantifies
the excess return per unit of deviation; it is defined as 
\begin{equation}
\frac{\mathrm{E}[R_\mathsf{A} -
    R_\mathsf{B}]}{\sqrt{\mathrm{Var}[R_\mathsf{A} -R_\mathsf{B}]}}, 
\end{equation}
where $R_\mathsf{A}$ and $R_\mathsf{B}$ are random variables
modeling the returns of assets $\mathsf{A}$ and $\mathsf{B}$,
respectively. A prominent variant of the Sharpe ratio,
called the \emph{Sortino ratio} (see~\cite{SorVDM91}), is relying on
the semideviation instead of the standard deviation of the returns
distribution. It is well-known and easily understood that the Sharpe
and Sortino ratios tend to give very different results for
highly-skewed return distributions. 

Finally, we would like to note that the concept of semivariance has
been generalized, resulting in the development of \emph{lower partial
  moments} by~\cite{Baw75} and~\cite{Fish77}. Essentially, the square
is replaced by an arbitrary power $a$ that can freely vary: 
\begin{equation}
\int_{-\infty}^{\tau} (\tau-r)^af(r)dr
\end{equation}
Varying $a$ might help in modeling the fact that an investor is more
(through larger values of $a$) or less (through smaller values of $a$)
sensitive to risk. In this paper, we have chosen to stick to $a=2$ for
simplicity reasons. In the following, we recall the concepts of
jump-diffusion models.  

\subsection{Jump-Diffusion Models}
\label{subsec:jumpmodels}
Jump-diffusion models are continuous-time stochastic processes
introduced in quantitative finance by~\cite{Mer76}, extending
the celebrated work of~\cite{BlaSch73} on option
pricing. These models are a mixture of a standard diffusion process
and a jump process. They are typically used to reproduce stylized
facts observed in asset price dynamics, such as mean-reversion and
jumps. Indeed, modeling an asset price as a standard
Brownian process implies that it is very unlikely that large jumps
over a short period might occur, as it is sometimes the case in real
life, unless for unrealistically large volatility values. Hence,
introducing the concept of jumps allows to take into account those
brutal price variations, which is especially useful when considering
risk management, for instance. 

Various specifications have been proposed in the literature and we
refer the reader to~\cite{ConTan04} for an extensive review. In what
follows, we consider first (in~\S\ref{sec:generalization}) the
standard jump-diffusion model with time invariant coefficients,
constant volatility and Gaussian distributed jumps. Later,
in~\S\ref{sec:jump_vol}, we will also consider more elaborated
stochastic processes, involving random jumps in returns and in
volatility. 

A basic jump-diffusion stochastic process is a
mixture of a standard Brownian process with constant drift $\mu$ and
volatility $\sigma$ and of a (statistically independent) compound
Poisson process with parameter $\lambda$ and whose jump size is
distributed according to an independent normal law $\mathcal{N}(\mu_Q,
\sigma^2_Q)$. More precisely, this model can be expressed as the
following stochastic differential equation: 
\begin{equation}
dX(t) = X(t) (\mu dt + \sigma dW(t) + J(t)dP(t)),
\label{eqn:process}
\end{equation}
where $X(t)$ denotes the process that describes the price of a financial
asset, with $\Pr[X(0) > 0] = 1$, where $\mu \in \mathbb{R}$ is the
process drift coefficient, $\sigma^2> 0$ is the process
variance, $W(t)$ is a standard Wiener process, $P(t)$ is a Poisson
process with constant intensity $\lambda > 0$ and $J (t)$ is the
process generating the jump size, that together with $P(t)$ forms a
compound Poisson process. The solution of the stochastic differential
equation~\eqref{eqn:process} is given by 
\begin{equation}
X(t) = X(0) e^{ \left(\mu-\frac{\sigma^2}{2}\right)t + \sigma W(t) +
  \sum_{k=1}^{P(t)}Q_k},
\label{eqn:process2}
\end{equation}
where $Q_k$ is implicitly defined according to $J (T_k) =
e^{Q_k}-1$, and $T_k$ is the time at which the $k$-th jump of the
Poisson process occurs. 
If $P(t) = 0$, the sum is zero by convention. We assume that the $Q_k$
form an independent and identically normally distributed sequence
with mean $\mu_Q$ and variance $\sigma^2_Q$.

The log-return of $X(t)$ over a $t$-period is defined as $Y_t =
\log X(t) - \log X(0)$ and, from~\eqref{eqn:process2}, its dynamic is
given by
\begin{equation}
Y_t = \left(\mu-\frac{\sigma^2}{2}\right) t + \sigma W(t) +
  \sum_{k=1}^{P(t)}Q_k,
\label{eqn:logret}
\end{equation}
The distribution of $Y_t$ is an infinite mixture of Gaussian
distributions 
\begin{equation*}
\mathcal{N}\left( \left(\mu-0.5\sigma^2\right) t + k\mu_Q, \sigma^2 t
  + k \sigma_Q^2) \right)
\end{equation*}
 and has a density function given by
\begin{equation}
f_{Y_t} (y) =  \sum_{k=0}^{+\infty} \left( \frac{e^{-\lambda  t}
    (\lambda t)^k}{k!} \frac{1}{\sqrt{2 \pi (\sigma^2 t + k
      \sigma_Q^2)}} e^{-\frac{1}{2}\frac{(y-((\mu-0.5\sigma^2) t +k
      \mu_Q )}{\sigma^2 t + k \sigma_Q^2}} \right) .
\label{eqn:dens0}
\end{equation}

\subsection{Contributions and Outline of this Paper}
\label{subsec:outline}
Our contributions in this paper can be summarized as follows: first of
all, we derive in~\S\ref{sec:explicitform} an explicit formula for
computing the semivariance of a standard jump-diffusion process when
the volatility is constant. To the best of our knowledge, it is the
first time that such a formula is
provided. Second, we propose in~\S\ref{sec:generalization} a
generalization of the~\cite{BalTor83,BalTor85} approximation of a
jump-diffusion process. Indeed, the simplification brought by Ball and
Torous is based on the fact that, during a sufficiently short time
period, and assuming a small jump intensity parameter, only a single
jump can occur. 

By doing so, the authors want to capture large and infrequent events
as opposed to frequent but small jumps. However, our analysis
in~\S\ref{sec:app} shows that limiting the jump intensity parameter
may result in an underestimation of the risk; hence, from a risk
management perspective, this approach does not seem to be
appropriate. It is the reason why we have preferred not to impose any
arbitrary condition on the Poisson
process intensity parameter and to estimate it by the maximum
likelihood method. Our extension of the work of Ball and Torous also
implies that more that one jump may occur during a single day. This is
a consequence of the fact that, when the intensity
parameter is sufficiently large, the probability of obtaining
more than one jump is then not negligible anymore. The only remaining
constraint that we keep in our approach is the fact that $\lambda$
should be smaller than a (large) upper bound. However, we show
in~\S\ref{subsec:charac} that this constraint is actually not a strong
limitation. Third, we apply our results in~\S\ref{sec:app} to compute
an estimation of the semivariance based on the Barclays US High Yield
Index returns, showing that the standard square-root-of-time rule
indeed may underestimate risk in certain periods and overestimate it
in other ones. For this, we make use of a customized optimization
algorithm based on differential evolution to maximize a likelihood
function. Last but not least, we discuss in~\S\ref{sec:jump_vol} the
extension of our work to a jump diffusion model with jumps in returns
and in volatility. Therein, we first recall the importance of
considering random jumps in returns and in volatility. Then, we
describe the stochastic volatility model that we use, which is an
extension of the model proposed by~\cite{EraJohPol03}. The statistical
estimation of its parameters is
addressed in the next section. We use in particular Markov Chain
  Monte Carlo (MCMC) methods to
derive their values. We propose in~\S\ref{sub:semi} a method to
compute an annualized semideviation once the model parameters have
been determined. Finally, we present some experimental results. 

\section{An Explicit Form for the Semivariance of a Jump-Diffusion
  Process}
\label{sec:explicitform}

We derive in this section an explicit formula for the semivariance
of a standard jump-diffusion model with time invariant coefficients,
constant volatility and Gaussian distributed jumps. To the best of our 
knowledge, this is the first time that an explicit formula is provided
for computing the semivariance.

In the following, let us denote respectively by $\phi(x)$ and
$\Phi(x)$ the probability density function and the cumulative
distribution function of a standard normal distribution
$\mathcal{N}(0, 1)$ with mean $0$ and variance $1$, i.e., 
\begin{equation*}
  \phi(x) = \frac{1}{\sqrt{2\pi}}e^{-\frac{x^2}{2}} \text{ and } 
  \Phi(x) = \int_{-\infty}^{x} \phi(t)\,dt \text{ for } x \in
  \mathbb{R}.
\end{equation*}
\begin{theorem}
  \label{the:semivariance}
  The semivariance of the density~\eqref{eqn:dens0} is given by
  \begin{equation}
    \sum_{k=0}^{+\infty} \frac{e^{-\lambda  t} (\lambda t)^k}{k!} \left(
      (D-\mu_k)^2 \Phi(D_k) + \sigma_k (D-\mu_k) f(D_k) +\sigma_k^2
      \Phi(D_k)  \right) ,
    \label{eqn:semivariance}
  \end{equation}
  where $\mu_k = \left(\mu - \frac{\sigma^2}{2}\right)t + k \mu_Q$, 
  $\sigma_k^2 = \sigma^2 t + k \mu_Q^2$, $D_k =
  \frac{D-\mu_k}{\sigma_k}$ and $k = 0,...,+\infty$.
\end{theorem}
The proof is given in Appendix~\ref{app:proof}. For a pure diffusion process without jump, the previous
formula~\eqref{eqn:semivariance} simplifies to the following one.
\begin{corollary}
The semivariance of a pure diffusion process with drift $\mu$ and
volatility $\sigma$ is equal to 
\begin{equation}
(D-\mu_0)^2 \Phi(D_0) + \sigma_0 (D-\mu_0) f(D_0) +\sigma_0^2 \Phi(D_0),
\label{eqn:pure}
\end{equation}
where $\mu_0 = \left(\mu - \frac{\sigma^2}{2}\right)t$, $\sigma_0^2 =
\sigma^2 t$ and $D_0 = \frac{D-\mu_0}{\sigma_0}$.
\end{corollary}
The proof is a direct consequence of Proposition~\ref{prop:prop0}
given in Appendix~\ref{app:proof} and of the fact
that~\eqref{eqn:logret} can be rewritten as 
\begin{equation}
Y_t = \left(\mu-\frac{\sigma^2}{2}\right) t + \sigma W(t)
\end{equation}
when we consider a pure diffusion process.
\section{Generalization of the Ball-Torous Approach}
\label{sec:generalization}

The task of fitting a jump-diffusion model to real-world data is not
as easy at it appears, and this fact has been early recognized, see
for instance~\cite{Bec81} or~\cite{Hon98}. Essentially, the reason
lies in the fact that the likelihood function of an infinite mixture
of distribution can be unbounded, hence resulting in
inconsistencies. However, by making some assumptions about the
parameters of this model, it is possible to accurately estimate
them. In the following, we present the approach
of~\cite{BalTor83,BalTor85}. 

Therein, the authors present a simplified version of a jump-diffusion
process by assuming that, if the jumps occurrence rate is small, then
during a sufficiently short time period only a single jump can
occur. Accordingly, for small values of $\lambda \Delta t$, $\Delta
P(t)$ can be approximated by a Bernoulli distribution of parameter
$\lambda\Delta t$, and the density of $\Delta Y(t)$ can then be
written as 
\begin{equation}
f_{\Delta Y} (y) = (1 -\lambda \Delta t) f_{\Delta D}(y) + \lambda
\Delta t ( f_{\Delta D} \star f_Q) (y),
\label{eqn:dens}
\end{equation}
where $f_{\Delta D}$ denotes the probability density function of the
diffusion part (including the drift), $f_Q$ the probability density function of the jump
intensity, and $\star$ denotes the convolution operator. As mentioned
in~\S\ref{subsec:jumpmodels}, $f_{\Delta D}$ follows
a normal law with mean $(\mu - \sigma^2/2) \Delta t $ and variance
$\sigma^2 \Delta t$. If $f_Q$ is distributed according to a normal
law statistically independent of the diffusion part, then the
convolution $f_{\Delta D} \star
f_Q$ of $f_{\Delta D}$ and $f_Q$ is normal with mean $(\mu -
\sigma^2/2) \Delta t + \mu_Q$ and variance $\sigma^2 \Delta t +
\sigma^2_Q$. 

For a sequence of observed log-returns $\Delta y_1,\dots, \Delta y_T$,
the log-likelihood $\log\mathcal{L}$ of the model parameters
$\bm{\theta} = (\lambda, \mu, \sigma, \mu_Q, \sigma_Q)^\mathsf{T}$
is obtained in a straightforward manner from~\eqref{eqn:dens} as
\begin{equation}
\log\mathcal{L}(\bm{\theta}\,|\,\Delta y_1, ..., \Delta y_T) (y) =
\sum_{t=1}^{T} \log f_{\Delta Y} (\Delta y_t\,|\,\bm{\theta}),
\label{eqn:loglikelihood}
\end{equation}
and the maximum likelihood estimator $\bm{\hat{\theta}}$ is obtained
by maximizing~\eqref{eqn:loglikelihood}. \cite{Kie78} has
shown that there may exist several local minima in such a mixture
setting, a fact that we have also observed in the experimental setup
that is the subject of the next section. 

While fitting the Ball-Torous model to real-world data
(see~\S\ref{sec:app} for more details and explanations about our
experimental setup) using~\eqref{eqn:loglikelihood}, we have
figured out that the assumption $\lambda\Delta t \ll 1$ might easily be
violated in practice. Indeed, we have observed on our data that, for
$\Delta t = \frac{1}{252}$ (i.e., $\Delta t$ representing one day in a
252-day trading year), the best obtained estimation for $\lambda$
ranged into the interval $[1, 252]$. This means that the value
$\lambda\Delta t$ was often nearer to $1$ than to $0$, and this
obviously questions the validity in practice of the assumption made
by Ball and Torous, at least in our experimental setup. In the
following, we propose a new methodology revolving around relaxing this
assumption to a milder one, namely that $\lambda\Delta t< 1$, and we
justify its use. 

\subsection{Our Methodology}
\label{subsec:our_methodology}
The methodology we describe in this section can be interpreted as an
extension of the work of~\cite{BalTor83}. In this
approach, the authors make the assumption that $\lambda \Delta t$ is
small, or in other words, that the expected number of jumps per
$\Delta t$ period is very small. However, in practice, this assumption
might not always be satisfied, as we observed it on our data. We
propose to relax this assumption and to replace it by the milder one
$\lambda \Delta t < 1$. For $\Delta t = 1/252$, assuming $252$ trading
days in a year, this translates to $\lambda < 252$: concretely, it
means that on average, there is no more than a single jump per day or,
equivalently, no more than 252 jumps per year. We easily agree that
this milder assumption may seem arbitrary at first sight. However, we
show in~\S\ref{subsec:charac} that it is not constraining at all,
since one can easily prove that a jump-diffusion process converges
in distribution to a pure diffusion process for increasing values of
$\lambda$. To summarize, with our methodology, we are able to fit any
jump diffusion without having any strong restriction on $\lambda$,
which is a major improvement in comparison to the work
of~\cite{BalTor83}.

Obviously, the reason why Ball and Torous decided to limit the jump
rate to a small value was to capture the apparition of brutal and rare
events. As a matter of fact, it is clearly more desirable on a
practical point of view to be able to model rare and large downside
market movements than frequent and small ones. This assumption implies
small $\lambda\Delta t$ values, and consequently, it means that the
occurrence of more than a single jump per $\Delta t$ period is
sufficiently unlikely that it can be neglected. 

However, we have experimentally figured out that the semivariance
computed on our relaxed model, i.e., allowing also frequent and small
jumps, may result in significantly higher values than on the model
assuming that $\lambda \Delta t$ is small (see
Figure~\ref{graph:sds-jump}). In other words, we observed that the Ball
and Torous model seems to underestimate the downside risk, compared to
our relaxed model, at least on our data set. A direct consequence of
not limiting the jumps occurrence rate $\lambda \Delta$ to a small
values is that the probability of having more than one jump during a
single day may become not negligible when $\lambda$ is large enough.
By allowing more than a single jump per $\Delta t$ period, one can
also take into account the fact that several bad news might influence
the market during a trading day, i.e., during a $\Delta t$
period. Obviously, if $\Delta t$ becomes sufficiently small, maybe as
small as a single second, it would be more difficult to justify in
practice several jumps during a time interval. However, when
discretizing the process in periods as large as a single day, we are
convinced that allowing more than a single jump better reflects the
reality. 

Consequently, in the following, we will assume that up to $m \geq 1$
jumps are possible during a time interval of $\Delta t$, where
$m\in\mathbb{N}$ is a \emph{finite} value. The probability
distribution of the possible number $k$ of jumps that
may occur during a time interval $\Delta t$ is then the following:
\begin{equation}
\label{eq:pk_values}
p_k = \frac{e^{-\lambda  \Delta t} (\lambda \Delta t)^k}{k!} \text{
  for } k = 0,...,m-1 \text{ and }
p_m = 1-\sum_{k=0}^{m-1} p_k.
\end{equation}
The resulting probability distribution, that we denote by
$\tilde{f}_{\Delta Y} (y)$ and which can be compared
to $f_{\Delta Y} (y)$ in~\eqref{eqn:dens}, can be expressed as 
\begin{equation*}
  \tilde{f}_{\Delta Y} (y) =  p_0f_{\Delta D}(y)+\sum_{k=1}^{m} p_k
  \left( f_{\Delta D} \star f_{Q^{(k)}}\right) (y),
\end{equation*}
where $f_{Q^{(k)}}$ denotes the convolution of $k$ density functions
$f_Q$. 

In what follows, we propose ourselves to look at the approximation
error of replacing a Poisson law by a truncated one. For a random
variable following a Poisson law of parameter $\lambda$, the
probability of obtaining a value strictly larger than $m$ is given by
the following function $f(\lambda)$:
\begin{equation}
\label{eqn:poisson}
  f(\lambda)  =   1 - \sum_{k=0}^{m} p_k\text{ and } p_k =
  \frac{e^{-\lambda} \lambda^k}{k!} \text{ with } k = 0,\dots,m.
\end{equation}
It is easy to see that $f(\lambda)$ is an increasing function in
$\lambda$. Indeed, the derivative of $f(\lambda)$ is given by
\begin{equation*}
f'(\lambda)  =   - \sum_{k=0}^{m} \frac{d\,p_k}{d\lambda}.
\end{equation*}
Then, we have $\frac{dp_k}{d\lambda} = - p_k +
p_{k-1}$ for $k = 0,\dots,m$, since 
\begin{equation*}
  \frac{dp_k}{d\lambda} = \frac{d}{d\lambda} \frac{e^{-\lambda}
    \lambda^k}{k!} = - \frac{e^{-\lambda}\lambda^{k}}{k!}  +
  \frac{ke^{-\lambda}\lambda^{k-1}}{k!} = -p_k + p_{k-1},
\end{equation*}
and where $p_{-1}$ is set to $0$ as
\begin{equation*}
\frac{dp_0}{d\lambda} = \frac{d}{d\lambda} e^{-\lambda} =
-e^{-\lambda} = - p_0.
\end{equation*}
Following this observation and telescoping the sum, the derivative
$f'(\lambda)$ can be rewritten as $f'(\lambda)   =   - \sum_{k=0}^{m}
\frac{dp_k}{d\lambda}  =   p_m > 0$. Now, if we substitute $\lambda$
by $\lambda \Delta t$ and assuming $\lambda \Delta t < 1$, then we can
observe that the supremum in~\eqref{eqn:poisson} is obtained for $\lambda  \Delta t= 1$. We
conclude that, assuming that $\Delta t < 1$, the probability of having
more than $m$ jumps is upper-bounded by the following expression:
\begin{equation}
\sum_{k=m+1}^{+\infty} \frac{e^{-1}}{k!} = 1 - \sum_{k=0}^{m}
\frac{e^{-1}}{k!}
\label{eqn:error}
\end{equation}
Table~\ref{table:bounds} gives a numerical upper bound for the
probability of obtaining values strictly larger than $m$ for different
values of $m$ when $\lambda \Delta t < 1$.
\begin{table}[t]
\begin{center}
\begin{tabular}{|l|c|r|}
  \hline
  $m$ & Upper Bound\\
  \hline
  1& 0.264 \\
  2& 0.080\\
  3& 0.019\\
  4& 0.003\\
  5& 0.001\\
  \hline
\end{tabular}
\end{center}
\caption{Upper bounds for the probabilities of obtaining more than $m$
  jumps for a Poisson law with $\lambda\Delta t < 1$.}
\label{table:bounds}
\end{table}
Based on the previous results, we conclude that ``truncating'' the
Poisson random variable combined with our assumption that $\lambda
\Delta t < 1$ gives a very good approximation of the return
distribution. Assuming that $t = n \Delta t$ and that no more than $m$
jumps can happen in a $\Delta t$ interval, we can derive a very accurate
approximation for the semivariance given
in~\eqref{eqn:semivarianceformula} by simply considering the first
$(mn+1)$ terms of~\eqref{eqn:semivariance} given in
Theorem~\ref{the:semivariance}:
\begin{equation}
\sum_{k=0}^{mn} \frac{e^{-\lambda  t} (\lambda t)^k}{k!} \left(
  (D-\mu_k)^2 \Phi(D_k) + \sigma_k (D-\mu_k) f(D_k) +\sigma_k^2
  \Phi(D_k)  \right).
\label{eqn:semivariance_trunc}
\end{equation}
For a Poisson process, the expected value
and the variance are equal to $\lambda$ for $t$ = 1 year. In other
words, if $\lambda < 252$ and assuming that $m = 5$, it means that we
consider the first $5 \times 252 + 1 = 1261$ terms of the
expression. Concretely, it means that we are ignoring events occurring
at more than $\frac{1260-252}{\sqrt{252}} \approx 63$ standard
deviations from the average of the Poisson process; a new time, this
implies that the approximation given in~\eqref{eqn:semivariance_trunc}
is an almost exact formula. 

\subsection{Discussion on the Assumption about $\lambda$}
\label{subsec:charac}
So far, we assumed that $\lambda \Delta t <
1$, a choice that might appear arbitrary at first sight. The purpose
of this section is to show that this assumption is actually not a
strong limitation. If we assume that $\Delta t = 1/252$ (1 day) and
that $\lambda \Delta t \geq 1$, or equivalently, that $\lambda \geq
252$, then its means that, on average, we have more than 1 jump a
day. In that case, the jump returns are small and often their order of
magnitude is less than the order of magnitude of a typical daily
return. In this particular situation, the effect of the jump returns
is hard to distinguish from the effect of the diffusion part of the
process. In other words, it is difficult to make the distinction between an
abnormal return coming from a jump or from the diffusion process. This
issue becomes even more evident when we have a look at the annualized
return distribution. When $P(t)$ is sufficiently large, then we can
invoke the central limit theorem to approximate $Q =
\sum_{k=1}^{P(t)}Q_k$, as the $Q_k$'s are iid random variables with
finite mean and variance. The mean and the variance of a compound
Poisson distribution derive in a simple way from the laws of total
expectation and of total variance. Formally, let us denote by
$\mathrm{E}_X[X]$ and $\mathrm{Var}[X]$ the expected value and the
variance of a random variable $X$, respectively. Furthermore, let
$\mathrm{E}_{X|Y}[X|Y]$ denote the conditional expectation of the
random variable $X$ conditioned by $Y$. Note that
$\mathrm{E}_{X|Y}[X|Y]$ is a random variable, and therefore, one
can compute its expected value. The law of total expectation tells us
that $E_Y\left[E_{X|Y}\left[X|Y\right]\right] = E_X[X]$, while the law
of total variance is formulated as $\mathrm{Var}\left[Y\right] =
  \mathrm{E}_X\left[\mathrm{Var}[Y|X]\right] +
  \mathrm{Var}_X\left[\mathrm{E}_{Y|X}[Y|X]\right]$.

Let us remind that in our jump-diffusion process, the amplitude $Q$ of a
single jump is assumed to follow a normal distribution with mean
$\mu_Q$ and standard deviation $\sigma_Q$, respectively, and that the
number of jumps $P$ in a given interval is modeled by a Poisson law
of parameter $\lambda \Delta t$. Finally, let us remind that $Q$ and
$P$ are independent random variables. We have
\begin{equation*}
\mathrm{E}[Q] = \mathrm{E}_P \left[ \mathrm{E}_{Q|P} [Q|P] \right] =
\mathrm{E}_P \left[P\cdot\mathrm{E}_Q [Q] \right] =
\mathrm{E}_P[P]\cdot \mathrm{E}_Q[Q] = \lambda t \mu_Q
\end{equation*}
as well as 
\begin{eqnarray*}
\mathrm{Var}[Q] &=& \mathrm{E} \left[ \mathrm{Var}_{Q|P} [Q] \right] +
\mathrm{Var} \left[ \mathrm{E}_{Q|P} [Q] \right]
= \mathrm{E} \left[ P\cdot \mathrm{Var}(Q) \right] + \mathrm{Var}
\left[P\cdot \mathrm{E} [Q] \right] \\
&=& \mathrm{E}[P]  \mathrm{Var}[Q] +  (\mathrm{E}[Q])^2
\mathrm{Var} [P]  
= \lambda t \sigma_Q^2 + \mu_Q^2 \lambda t = \lambda t ( \sigma_Q^2 +
\mu_Q^2)
\end{eqnarray*}
Based on the previous development, we can conclude that the jump
diffusion process converges in distribution to a normal distribution
when $P(t)$ becomes large:
\begin{equation}
Y_t \sim \mathcal{N}\left( (\mu-\frac{\sigma^2}{2}+ \lambda \mu_Q) t,  (\sigma^2 +
  \lambda(\sigma_Q^2+\mu_Q^2))t \right) .
\label{eqn:logret2}
\end{equation}
In short, the annualized return distribution converges to the return
distribution that we would obtain for a pure diffusion process but
with different drift and volatility parameters. Concretely, when
$\lambda \Delta t$ is larger than 1, then the interest of such a model
is limited and the use of a pure diffusion process is preferable at
least in our context.  



\section{Applications}
\label{sec:app}
In order to practically illustrate our results, we have chosen to
focus on the \emph{Barclays US High Yield Bond Index}
between January 3rd, 2007 and July 31st, 2012. We explain in Appendix~\ref{app:NAV}
why we can use a jump-diffusion process to model a bond benchmark index.
The daily log-returns (see Figure~\ref{fig:barclays}) of this index form a good example of a
highly-skewed, long-tailed distribution with a negative sample
skewness of $-1.63$ and a sample kurtosis of $24.08$. An histogram
approximating the probability density function of the log-returns, as
well as a normal law whose parameters are the log-returns sample mean
$\hat{\mu} \approx 0.033\%$ and standard deviation $\hat{\sigma}
\approx 0.42\%$ are depicted in Figure~\ref{fig:histogram}.  

The semivariance of the log-returns can easily be computed with help
of~\eqref{eq:TSV}. However, it is not clear at all how we should
proceed to annualize this semivariance, except in one case: if we
assume that the log-returns follow a pure diffusion process, with no
drift, and if the threshold $\tau$ is set to zero, then the variance
of the process increases linearly with time, and the annualized
semivariance is just the daily semivariance multiplied by the square
root of time. However, as soon as we introduce jumps in the stochastic
process, there is no reason to apply this rule anymore. Indeed, we
have derived in~\eqref{eqn:semivariance} an exact formula that
enables us to accurately annualize a daily semivariance into an
annualized semivariance. 

The purpose of this section consists in performing a rolling analysis
and to compare the semivariances that we obtain from three
distinct methodologies:
\begin{enumerate}
\item an approach relying on fitting our data to a jump-diffusion
  process, and computing an annualized semivariance according to
  our generalized Ball-Torous model
  and~\eqref{eqn:semivariance_trunc}, which we have
  shown to be an almost exact approximation
  of~\eqref{eqn:semivariance};
\item an approach relying on fitting our data to a pure diffusion
  process (i.e., without jumps), and computing an annualized
  semivariance according to~\eqref{eqn:pure};
\item and the square-root-of-time rule.
\end{enumerate}

The objective of this analysis is to experimentally determine which
one of these three methodologies seems to provide the best results. In
our computations, all the semivariances are based on 1 year of
historical data, that is, 252 points, and we perform a rolling
analysis over the whole period. We also have decided to set $\tau=0$
in all of our tests, since the use of the square-root-of-time is
only justified for this particular value.\footnote{Note that the
  formulas derived in \S\ref{sec:explicitform} and
  \S\ref{sec:generalization} can be applied to any value of $\tau$.}
The semivariance in $t$ is computed as follows: first, we compute the 
daily semivariance based on historical data from $t-251$ days to $t$;
then, we transform the daily semivariance to obtain an annualized
semivariance by replacing $t$ by $\Delta t$
in~\eqref{eqn:semivariance_trunc} and~\eqref{eqn:pure},
respectively, or by applying the square-root-of-time rule. 

This section is organized as follows: in~\S\ref{subsec:DEA}, we
discuss our parameter fitting procedure, that rely on the use of a
differential evolution algorithm. We first recall the salient features
of this type of optimization procedures, and we describe the
difficulties encountered, mainly in terms of stability of the obtained
results. Then, in~\S\ref{subsec:discussion_results}, we describe and
discuss the experimental results obtained while computing annualized
semivariances according to the three different approaches quickly
described above. 
\begin{center}
  \begin{figure}
    \includegraphics[width=0.9\linewidth]{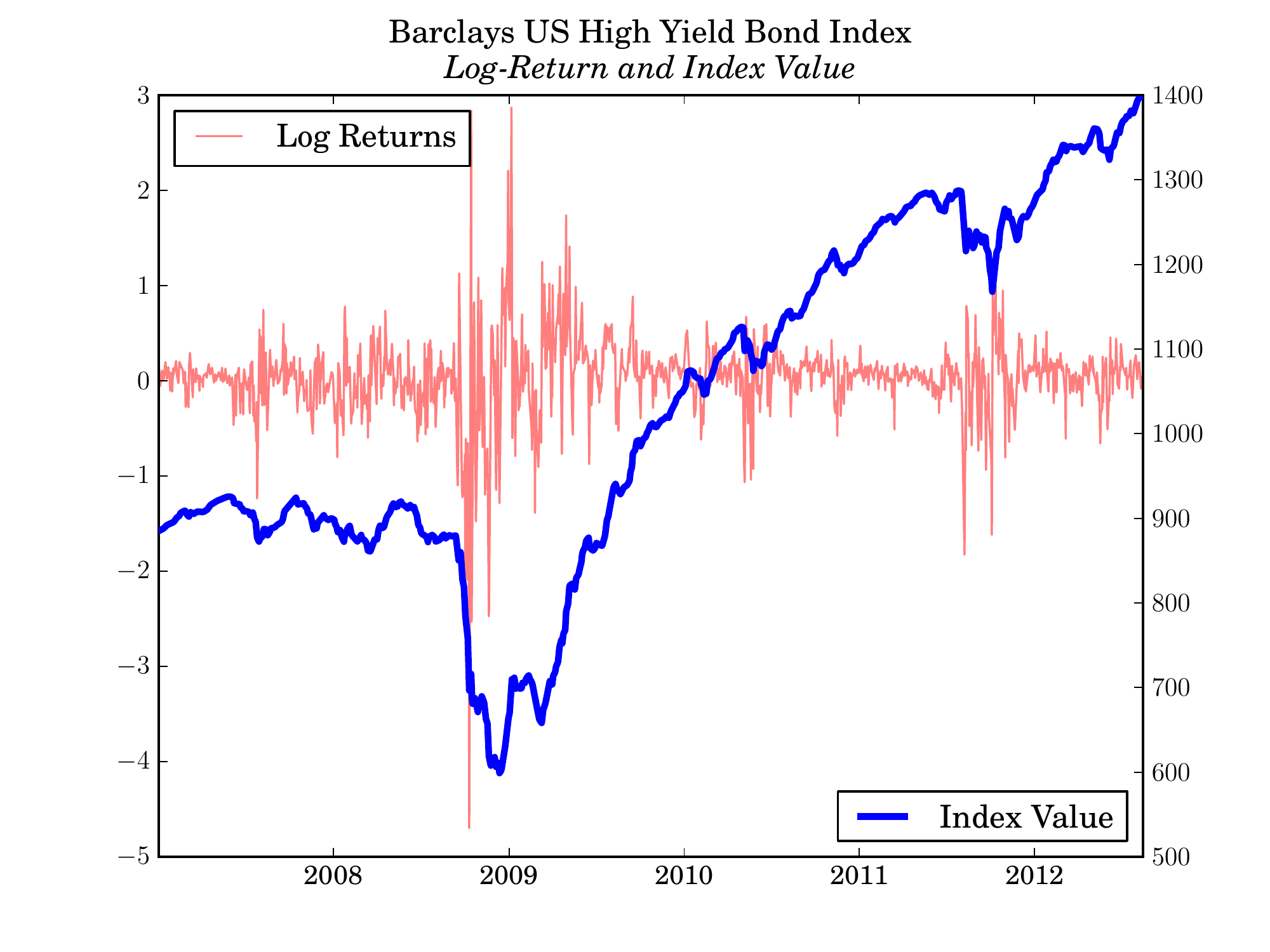}
    \caption{\label{fig:barclays}Barclays US High Yield Bond Index --
      Daily log-returns (in \%) and index value from January 3rd, 2007 and July
      31st, 2012.}
  \end{figure}
\end{center}
\begin{center}
  \begin{figure}
    \includegraphics[width=0.9\linewidth]{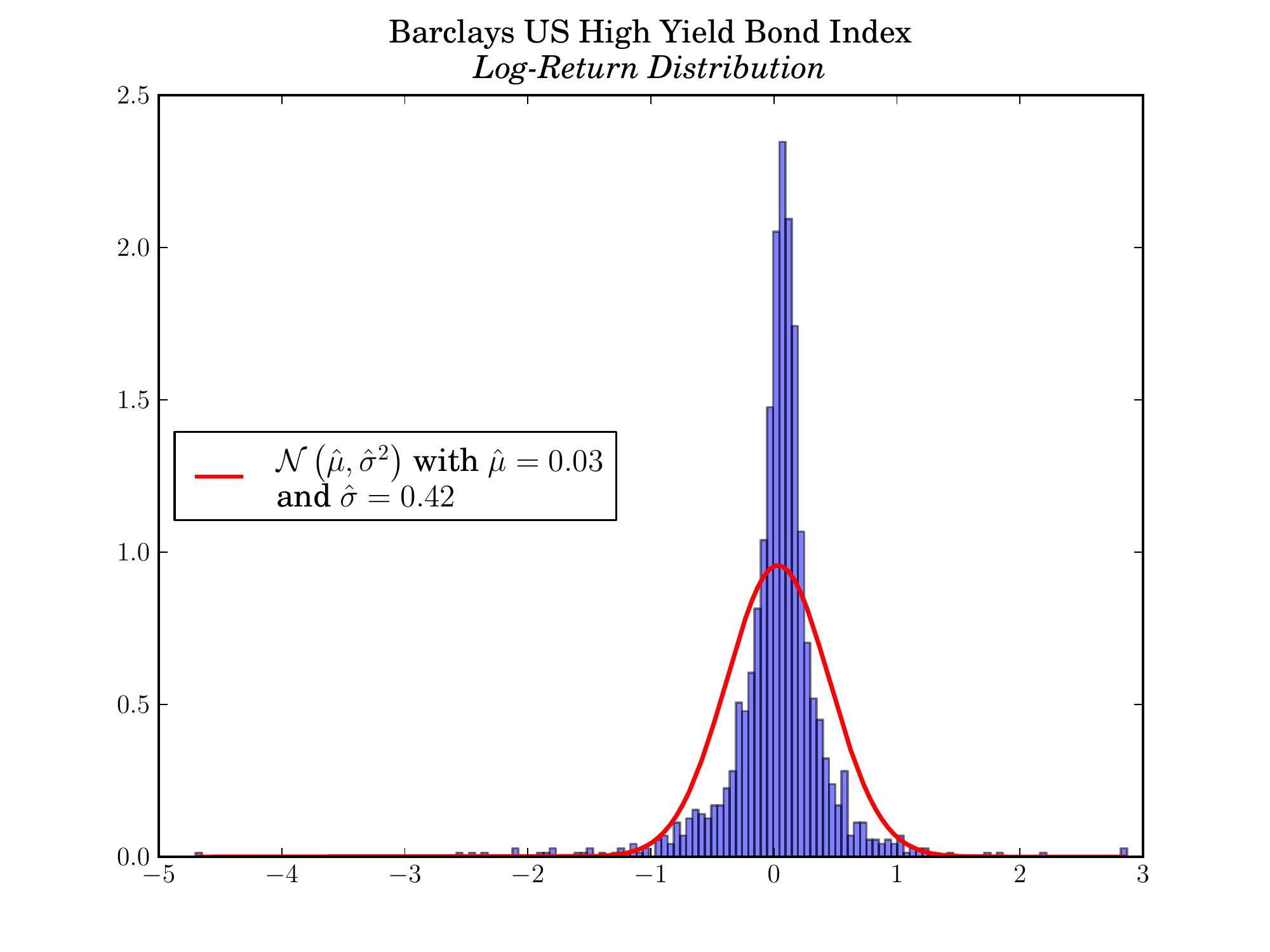}
    \caption{\label{fig:histogram}Histogram of the daily returns. The
      returns show negative sample skewness equal to $-1.63$ and a
      sample kurtosis equal to $24.08$. The sample mean $\hat{\mu}$
      and standard deviation $\hat{\sigma}^2$ are in \%.}
  \end{figure}
\end{center}

\subsection{Optimization with Differential Evolution}
\label{subsec:DEA}

The main challenge that we faced from an optimization point of view is
the fact the objective function~\eqref{eqn:loglikelihood} possesses
several local optima. To overcome this issue, we have decided to use
the optimizer described in~\cite{ArdMul13}. The main motivations
for us to choose it is its easiness to use, its capacity to handle
local optima and its ability to manage constraints about the
parameters. Moreover, it was successfully tested in~\cite{ArdOspGir11b}
on diffusion problems. Finally, this kind of algorithm has proved to
be flexible enough for us to customize it for the rolling analyses
that we are dealing with in this paper. We call the variant we will
describe later \emph{differential evolution algorithm with memory}. 

Optimization methods based on Differential Evolution (DE) is a 
class of search algorithms introduced by~\cite{StoPri97} that belongs
to the class of evolutionary
algorithms. These algorithms assume no special mathematical property
about the target function, like differentiability for gradient
descents or quasi-Newton methods. It means that DE can handle
non-continuous or very ill-conditioned optimization problems by
searching a very large number of candidate solutions; we note that
this number of candidate solutions can be fitted to the computational
power one has at disposal. The price to pay on this lack of assumption
is that optimization algorithms relying on DE are never guaranteed to
converge towards an optimal solution. However, such meta-heuristics
still seem to deliver good performances on continuous problems, see
e.g.~\cite{PriStoLam05}. 

Roughly speaking, DE-based optimization algorithms exploit
bio-inspired operations such as crossover, mutation and selection of
candidate populations in order to optimize an objective function over
several successive generations. Let $f: \mathbb{R}^d\longrightarrow
\mathbb{R}$ denote the $d$-dimensional \emph{objective function} that
must be optimized. For the sake of simplicity, we assume that $f$ has
to be minimized; this is not restrictive, as $f$ can easily be
maximized by minimizing $-f$. The DE algorithm begins with a
population of $v$ initial vectors $\bm{x}_i \in \mathbb{R}^d$, with $1
\leq i \leq v$, that can either be randomly chosen or provided by the
user. For a certain number of rounds, fixed in advance, one iterates
then the following operations: for each vector $\bm{x}_i \in
\mathbb{R}^d$, one selects at random three other vectors $\bm{x}_a$,
$\bm{x}_b$ and $\bm{x}_c$, with $1 \leq a, b, c, i \leq t$ being all
distinct, as well as a random index $1 \leq \rho \leq d$. Then,
the successor $\bm{x}_i^\prime$ of $\bm{x}_i$ is computed as follows:
given a \emph{crossover probability} of $\pi \in [0, 1]$, a
\emph{differential weight} $\omega \in [0, 1]$ and $\bm{x_i} =
(x_{i,1}, \dots, x_{i,d})^\mathsf{T}$, for each $1 \leq j \leq d$, one draws a
number $r_j$ uniformly at random. If $r_j  < \pi$, or if $j = \rho$,
then one sets $x_{i,j}^\prime = x_{a, j} + \omega\cdot(x_{b, j} -
x_{c,j})$; otherwise, $x_{i, j}^\prime = x_{i,j}$ keeps
untouched. Once this mutation done, the resulting vector
$\bm{x}^\prime$ is evaluated with respect to the objective function,
and replaces $\bm{x}_i$ in the candidate population if it is
better. Otherwise, $\bm{x}^\prime_i$ is rejected. For more details, we
refer the reader to~\cite{PriStoLam05} and~\cite{StoPri97}.
For practical purposes, several implementations of DE-based
optimization algorithms are currently available, as illustrated by the
web-based list of DE-based software for general purpose optimization
maintained by~\cite{storn-website}. In what follows, we rely on
the package \texttt{DEoptim},
see~\cite{ArdMul13,MulArdGilWinCli11,ArdBouCarMulPet11}, which
implements DE-based optimization in the R language and environment for
statistical computing, see \cite{r-lang}. 

For the optimizer, we have used a maximum number of iterations of
$250$, default values for the crossover probability of $\pi = 0.5$ and
for the differential weight $\omega = 0.8$,
as well as an initial population of $v=200$ vectors. We also have
constrained the value of $\lambda$ to the interval $[0,252]$ to be
compliant with the assumption $\lambda \Delta t < 1$. The main
challenge that we have faced for our rolling analysis is the lack of
stability over time of some parameters of the objective function. This
is illustrated in the upper graph of Figure~\ref{graph:stab1}, that
depicts the values computed by the DE-based optimizer for
$\lambda$.\
\begin{center}
  \begin{figure}
    \begin{center}
      \includegraphics[width=\linewidth]{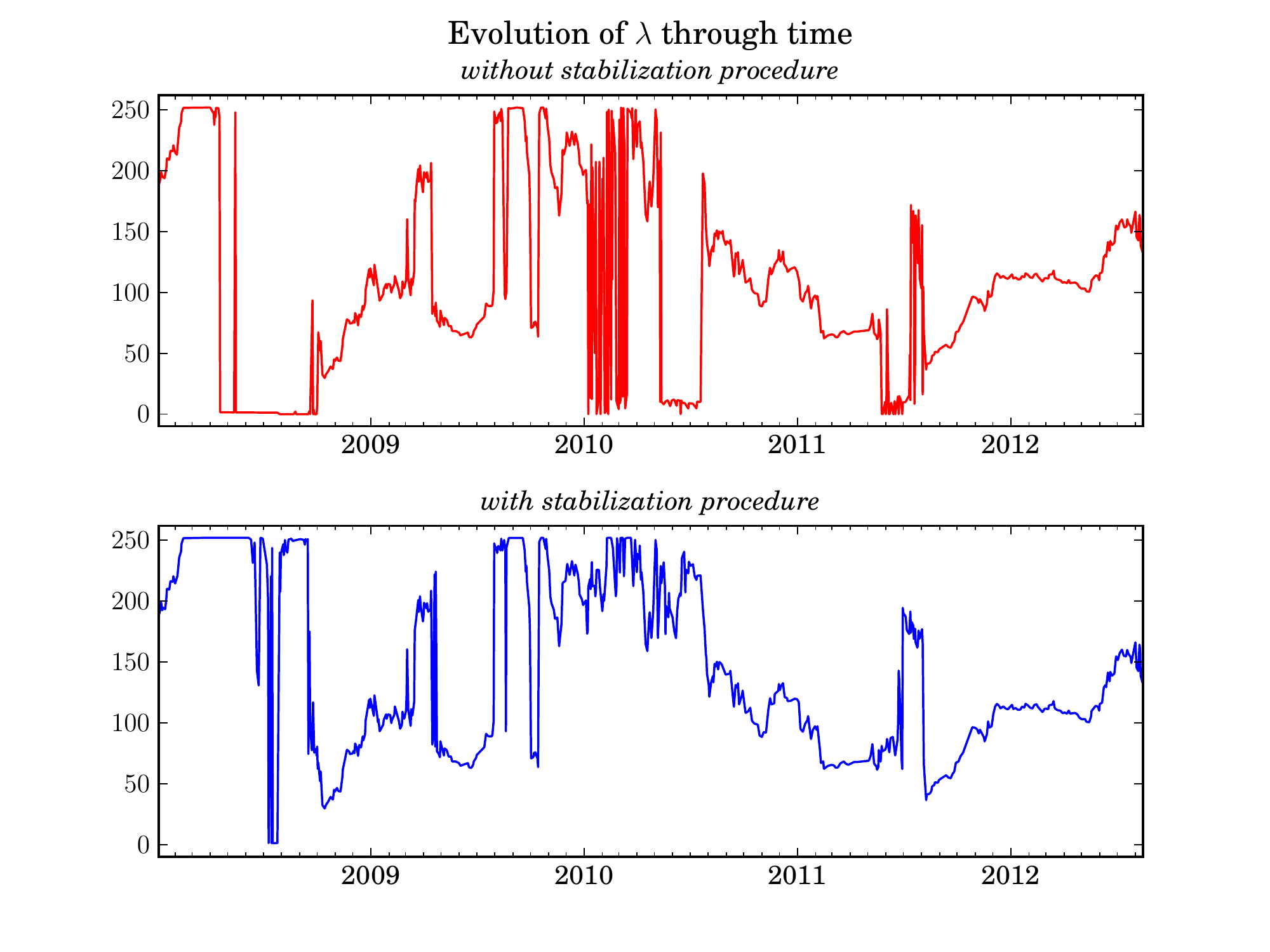}
    \end{center}
    \caption{Evolution of $\lambda$ without and with the stabilization
      procedure. The stabilization procedure relies on feeding the initial
      population with the 50 last solutions, if available.}
    \label{graph:stab1}
  \end{figure}
\end{center}

We can observe that this lack of stability can be extreme for certain
periods of time. For instance, in the first half of 2010, we can see
that $\lambda$ jumps from a high value near to its upper bound $252$
to a very low value the next date. This is due to the fact that
objective function possesses several optima and that for
each of these solutions, $\lambda$ can be very different. 

In order to overcome this issue, we propose to add ``memory''
to the algorithm. Indeed, for each date, we compute a solution that
maximizes the log-likelihood~\eqref{eqn:loglikelihood}. Our idea
consists in storing these solutions and to feed the initial population
with the last 50 solutions of the optimization problem. Concretely, to
determine the solution at a specific date, the initial population --
whose size is $v=200$ -- is fed with the solutions of the
optimization problem at times $t-1,t-2,\ldots,t-50$. There is no reason
why the solution at time $t$ is the same as time $t-1$, except if the
objective function remain the same; however, it is very likely that
the solution at $t-1$ will provide a good starting point for the
current optimization problem if the objective function has changed,
and directly an optimal solution if the objective function has not
changed. With this procedure, we also have the guarantee that the
solution for the date $t$ will not switch from one point to another
with the same log-likelihood value at the next optimization problem of
date $t+1$ if the objective function has not changed. The effect of
this stabilization procedure are illustrated in the lower graph of
Figure~\ref{graph:stab1}. The stabilization procedure is very important
for the interpretation of our results. Without this stabilization
procedure, it would be far more difficult and even impossible to
interpret the results for some periods of time.  

\subsection{Discussion of our Results}
\label{subsec:discussion_results}

\begin{center}
  \begin{figure}
    \begin{center}
      \includegraphics[width=0.9\linewidth]{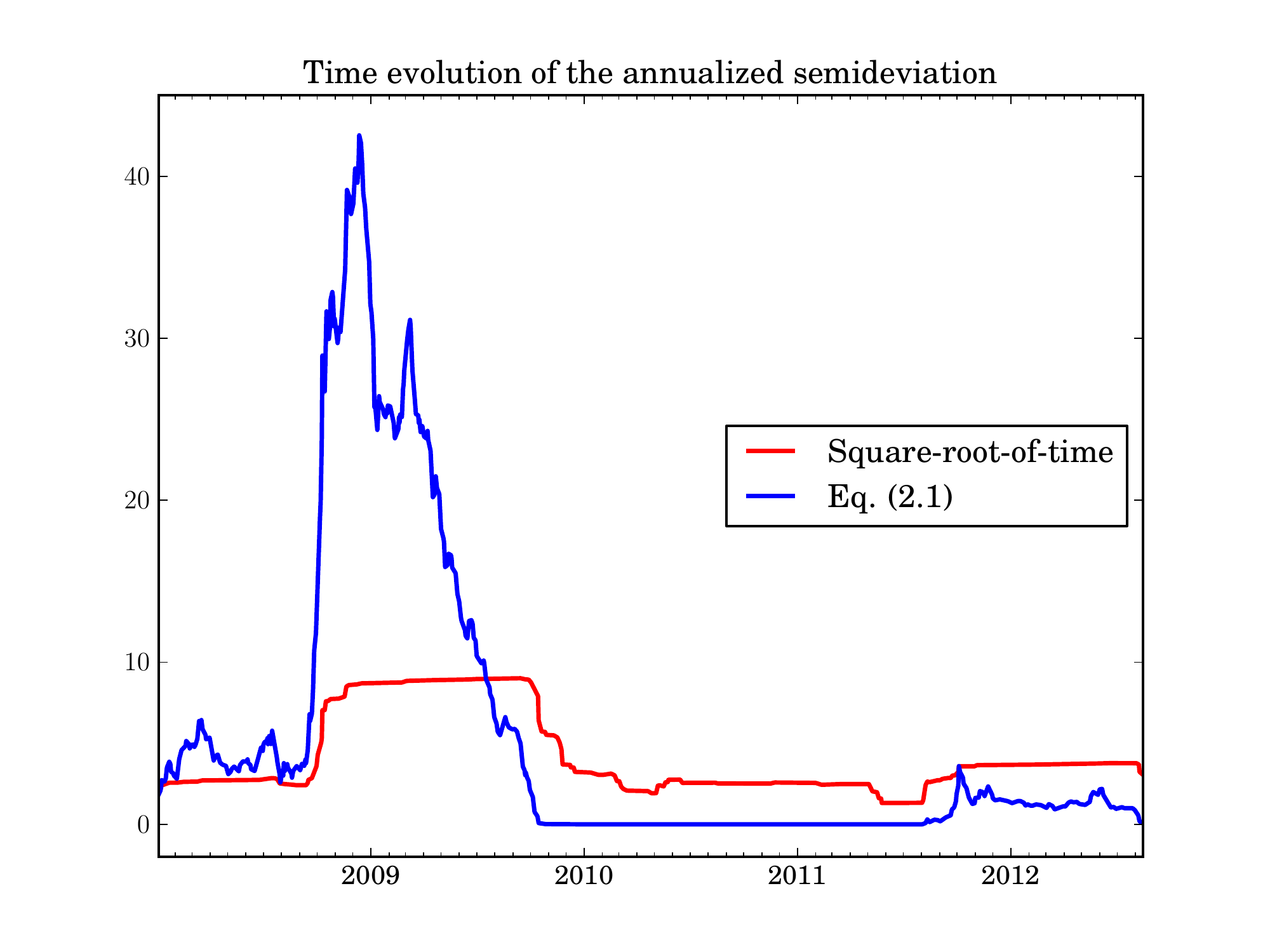}
    \end{center}
    \caption{Comparison of the semideviation (in \%) obtained via the
      square-root-of-time rule with the semideviation (in \%) based on a
      jump-diffusion process. We can observe that the semideviation
      based on the square-root-of-time rule significantly underestimates
      the risk in periods of high volatility and overestimates the
      risk in periods of lower stress.}
    \label{graph:res}
  \end{figure} 
\end{center}
We start by comparing in Figure~\ref{graph:res} the evolution
of the annualized semideviation, i.e., the square root of the
semivariance, computed thanks to the square-root-of-time rule to the 
semideviation computed through fitting a jump-diffusion process to the
data under the constraint $\lambda \leq 252$ on which we applied our
formula given in~\eqref{eqn:semivariance_trunc}. It is quite
striking to see that the risk estimated on the jump-diffusion process
is up to 4 times larger than the one computed thanks to the
square-root-of-time rule. At the beginning of 2009, the first one is
larger than $40\%$ while the other is just below $10\%$. The risk is
not only underestimated in period of crisis, but also seems to be
overestimated when the market rallies or in period of low or medium
stress. Indeed, we can observe that the jump-diffusion semideviation
is almost zero from the end of 2009 to the end of 2011, while the
square-root-of-time semideviation remains at a level around $2\%$
during that period.
\begin{center}
  \begin{figure}
    \begin{center}
      \includegraphics[width=0.9\linewidth]{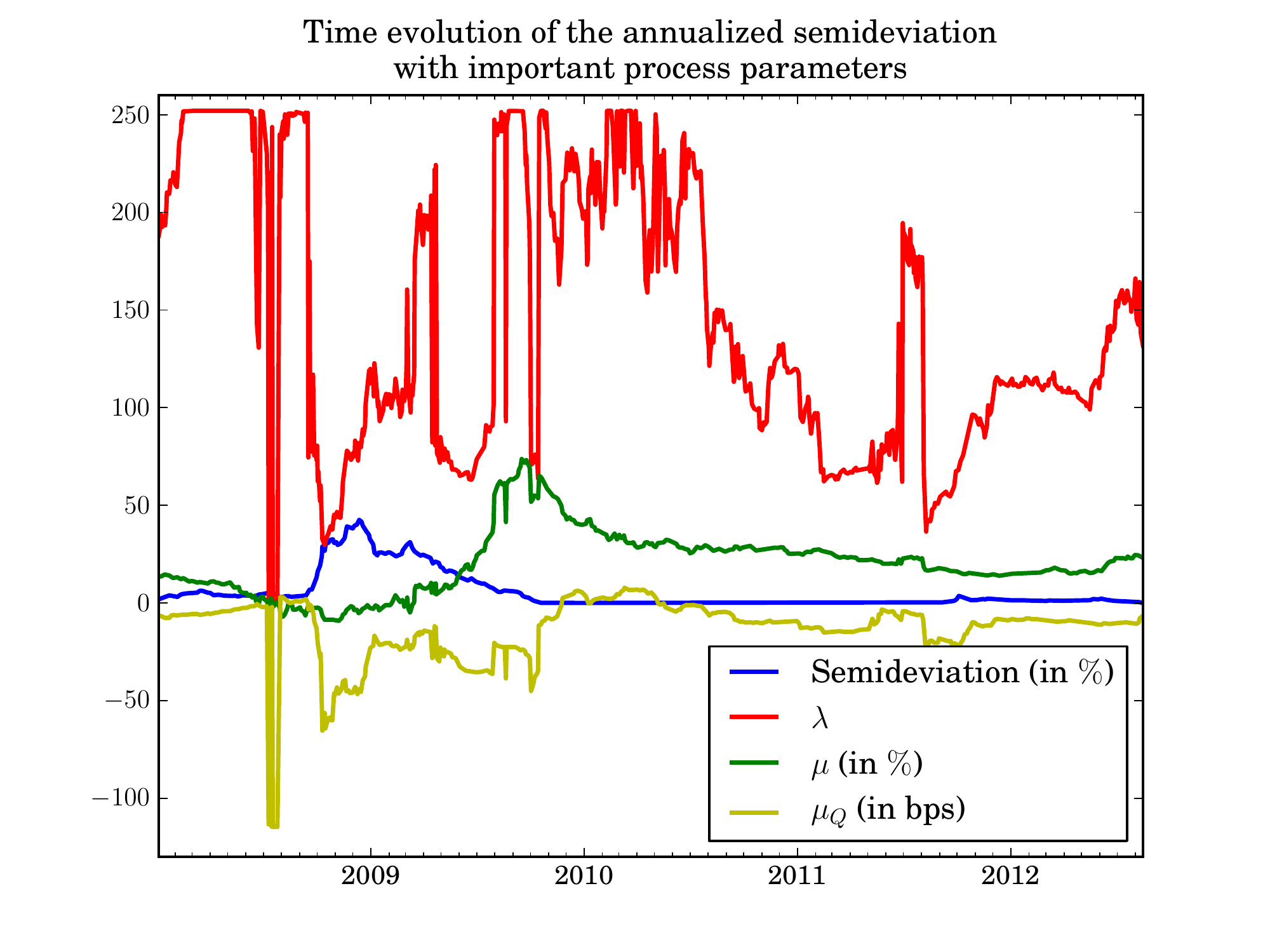}
      \caption{Time evolution of the jump-diffusion process parameters}
      \label{graph:params-interpretation}
    \end{center}
  \end{figure}
\end{center}

\subsubsection{Relationship between Important Process Parameters}

To better understand how the jump-diffusion semideviation works, it is
very important to understand the relationship between the different
process parameters. Figure~\ref{graph:params-interpretation}
illustrates the evolution through time of some important parameters
of the jump-diffusion process, namely $\lambda$, $\mu$ and $\mu_Q$. 
We can observe a direct relationship between $\lambda$ and
$\mu_Q$, namely that when $\lambda$ gets small, then $\mu_Q$ gets
large (in absolute value). This behavior is quite easy to interpret:
in normal market conditions, $\lambda$ is rather high, while $\mu_Q$
is rather low. The combination of these two parameters under normal
market conditions captures a well-known effect in credit portfolios, a
negative asymmetry in their return distribution. 

The parameter $\mu$ has an offsetting effect counterbalancing the
compound Poisson process. For example, let us assume that $\lambda =
252$ and that $\mu_Q$ has an average size of 1 basis point. In that
case, the cumulative annual impact of the accidents driven by the
Poisson process is on average a negative return of $252 \times 1 =
252$ basis points. Let us imagine that the annual return that we can
expect from our investment is $7\%$, then it is very likely that the
statistical estimation of $\mu$ would be
around $10\%$ (note that from~\eqref{eqn:logret}, we also have to
take into account the role of $\sigma$ to determine the annual
expected return). In crisis situation, $\lambda$ really captures the
frequency of extreme events and its value typically drops from 252 to
smaller values, sometimes close to 1. By contrast, $\mu$ is by
construction less sensitive to these extreme accidents, even though we
can see that it is impacted, by looking at
Figure~\ref{graph:params-interpretation}. Indeed, we can observe that
$\mu$ has dropped to negative values during the 2008 financial crisis. 

It is also quite interesting to have a look at what happened during
the rally phase that started in March 2009. We can observe that $\mu$
has peaked to values larger than $60\%$. It is worth
mentioning that even though the value $\mu$ has been very large
at the end of 2009 due to a strong rally in the market, $\mu_Q$ stayed
negative during that phase. It is only in year 2010 that $\mu_Q$ had
positive values. This situation is quite exceptional, as we are used
to negative skewness in the returns for credit portfolios. However,
under some very particular circumstances, they can exhibit 
positive skewness and the values of $\mu_Q$ cannot be interpreted any
more as losses, but as gains. 

Note also that the jump-diffusion semideviation can be very close to
zero in very special cases, and in particular when $\mu$ is much
larger than $\lambda \cdot \mu_Q$, when $\mu_Q$ is negative or in
the situation where both $\mu$ and $\mu_Q$ are positive. 
This just means that in this kind of situation, the probability of
obtaining a annualized return below zero is very low, resulting in a
semideviation close to zero. 

The two worst crisis periods that the financial markets have
experienced during the last five years are, by order of importance,
the 2008 financial crisis and the US Debt Ceiling crisis in
2011. After weeks of negotiation, President Obama signed the Budget
Control Act of 2011 into law on August 2, the date estimated by the
department of the Treasury where the borrowing authority of the US
would be exhausted. Four days later, on August 5, the credit-rating
agency Standard \& Poor's downgraded the credit rating of US
government bond for the first time in the country's history. Markets
around the world experienced their most volatile week since the 2008
financial crisis with the Dow Jones Industrial Average plunging for
635 points (or $5.6\%$) in a single day. However, the yields from the
US Treasuries dropped as investors were more concerned by the European
sovereign-debt crisis and by the economic prospect for the world
economy and fled into the safety of US government bonds. This is a
reason why we could see the jump-diffusion semideviation increasing
substantially during that period.

\subsubsection{Relationship between $\lambda$ and the Semideviation}
It is worth studying the relationship between $\lambda$ and the
semideviation of the jump-diffusion process. In Figure~\ref{graph:sds-jump}, we give the
evolution of the semideviation based on different values of $\lambda$. 
\begin{center}
  \begin{figure}
    \begin{center}
      \includegraphics[width=0.9\linewidth]{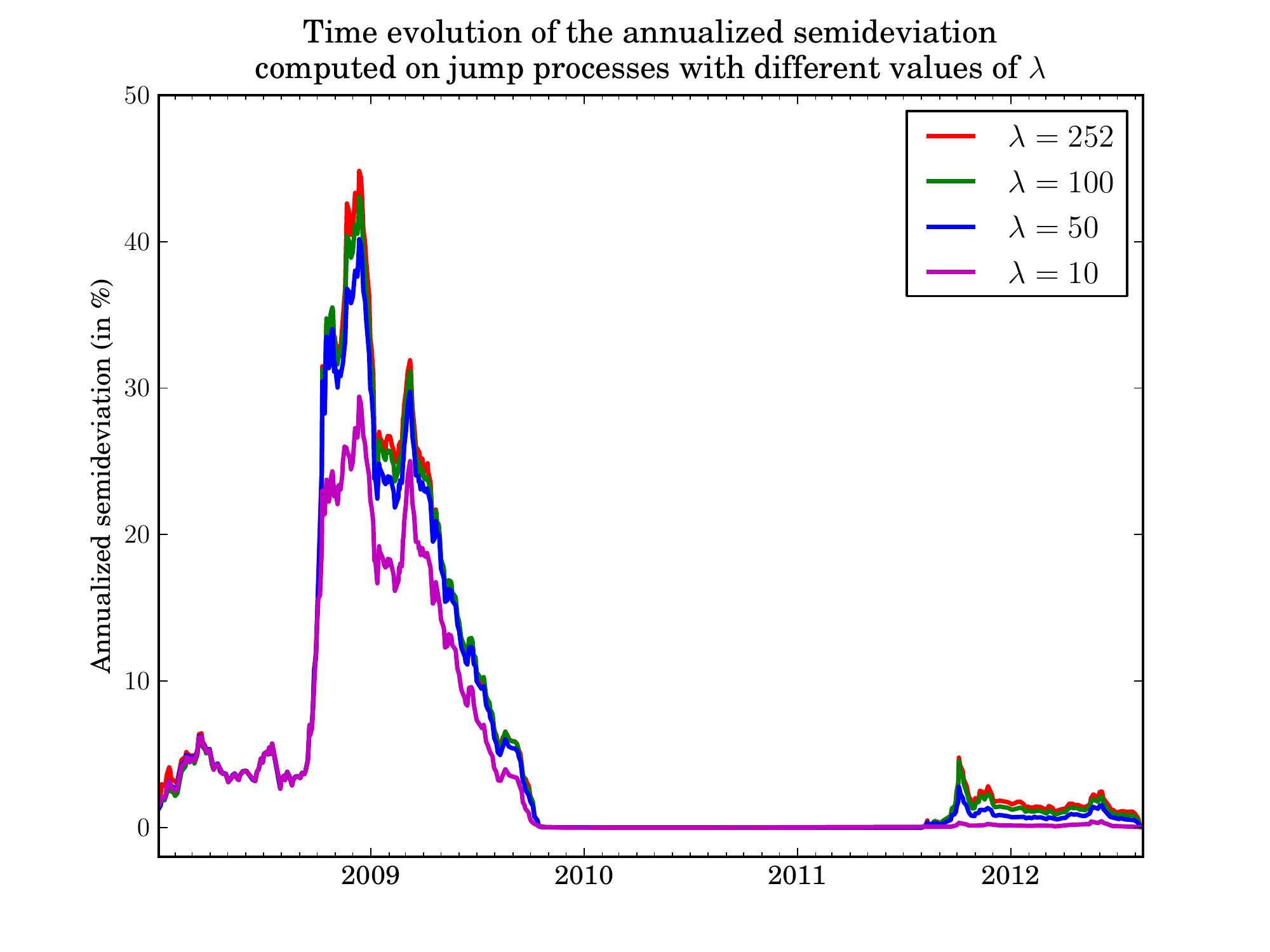}
      \caption{\label{graph:sds-jump}Evolution of the semideviations
        (in \%) for different values of
        $\lambda$. The curve in red corresponds to the constraint $\lambda =
        252$, in green to $\lambda =  100$, in blue to $\lambda =  50$ and
        in violet to $\lambda = 10$. Note that the curves in red and in
        green are almost identical.}
    \end{center}
  \end{figure}
\end{center}

This analysis clearly shows that the assumption made about $\lambda$
has a huge impact on the computation of the semideviation.
In particular, it has been a standard practice since the publication
of the work of~\cite{BalTor83} to use a Bernoulli
jump diffusion process as an approximation of the true diffusion
process by making the assumption that $\lambda \Delta t$ is
small. However, our analysis in Figure~\ref{graph:res} indicates that a
model with a larger $\lambda$ fits best the data. Obviously, the
motivation of Ball and Torous was to capture large and infrequent
events by opposition to small and frequent jumps. However, imposing to
$\lambda \Delta t$ to be small has the undesired consequence of
underestimating the risk in periods of high volatility. Indeed, the
graph in Figure~\ref{graph:sds-jump} shows that the semideviation
obtained with $\lambda = 252$ can be as much as 50 \% larger in a
period of stress than the semideviation computed with the constraint
that $\lambda = 10$. This observation was the main motivation for us to
extend the work of Ball and Torous by allowing models where the
frequency of the jumps is not imposed anymore, but driven by the
data. 

In \S\ref{subsec:charac}, we have shown that when $\lambda$ is large,
then the jump-diffusion return distribution should be close to the
distribution of a pure diffusion process but with different
parameters. Rather than superposing the semideviation computed out of
a jump-diffusion process with $\lambda = 252$ with a pure diffusion
process, we have put in Figure~\ref{graph:sds-jump-pure} the
jump-diffusion semideviation and its \emph{difference} with the
semideviation obtained from a pure diffusion process. 
\begin{center}
  \begin{figure}
    \begin{center}
      \includegraphics[width=0.9\linewidth]{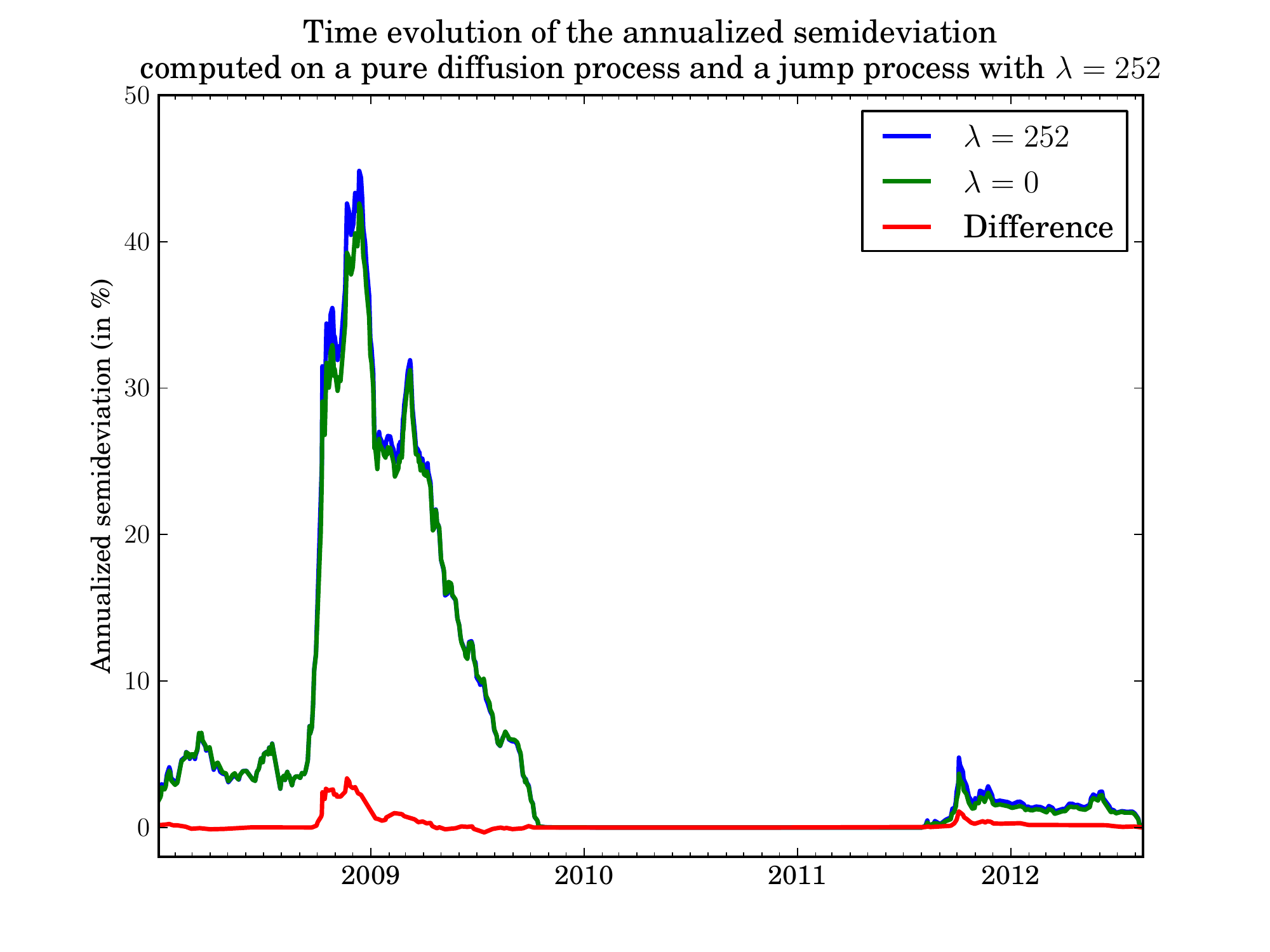}
    \end{center}
    \caption{\label{graph:sds-jump-pure}Comparison between the
      semideviation (in \%) obtained from a jump
      diffusion process with $\lambda = 252$ with a pure diffusion
      process with $\lambda=0$. The blue and green lines represent the
      jump-diffusion and pure diffusion semideviations, respectively, 
      while the red one is the difference between the two.
      Even though the magnitude of the semideviations are
      similar, we can observe that the difference is substantial
      during the 2008 financial crisis. }
  \end{figure}
\end{center}

The difference is small, except during the 2008 financial crisis. During
that period, the difference can be larger by several \%. The central
limit theorem tells us that the annualized return distribution should
converge to a normal distribution whose parameters are given
in~\eqref{eqn:logret2}. 
But the parameters of the pure diffusion process have been determined
by maximization of the log-likelihood and may give results quite
different from the ones in~\eqref{eqn:logret2}. Moreover, we can see
that the two semideviations are diverging in period of crisis, i.e., when it is
difficult to calibrate the models due to the appearance of extreme
events that make the determination of the parameters of the models
quite challenging. 

To illustrate this fact, we have depicted the
evolution of the log-likelihood value for both processes on
Figure~\ref{graph:likelihood}. We have also compared the
jump-diffusion and the Bernoulli jump-diffusion semideviations which
should be the same based on the results
from~\S\ref{subsec:charac}. Our results are displayed on
Figure~\ref{graph:sds-jump-bern}. The difference is negligible over
all the period, meaning that our methodology provides the same results
as the one based on the model of Ball and Torous provided $\lambda$ is
constrained to be small. 
\begin{center}
  \begin{figure}
    \begin{center}
      \includegraphics[width=0.9\linewidth]{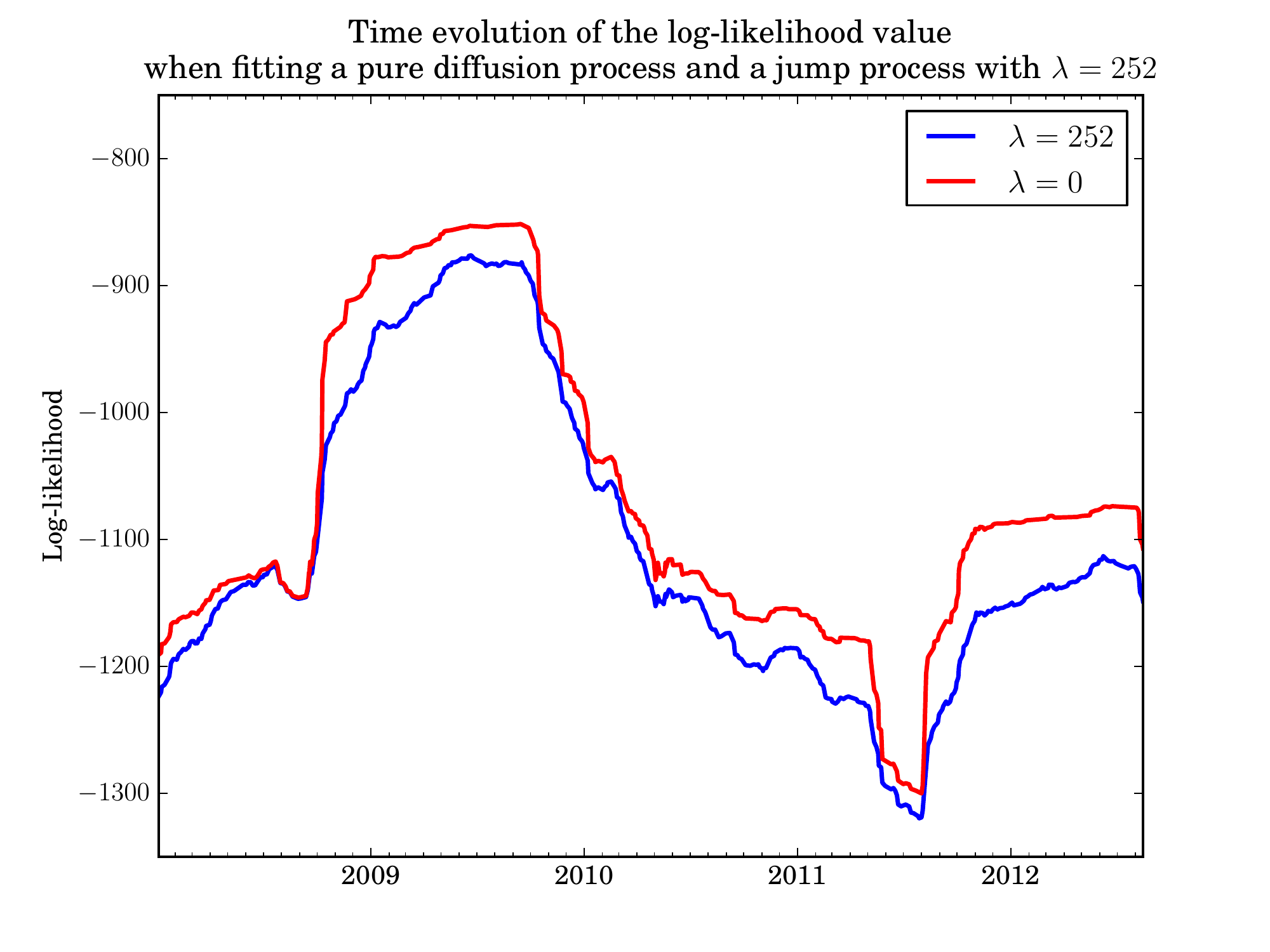}
    \caption{\label{graph:likelihood}Comparison of the log-likelihood
      for the jump-diffusion
      process (in blue) with $\lambda = 252$ and with the pure diffusion
      process (in red). In both cases, we can see that the log-likelihood
      is smaller in periods of high market stress, which means that the
      calibration of the model is also more challenging.}
     \end{center}
  \end{figure}
\end{center}
\begin{center}
  \begin{figure}
    \begin{center}
      \includegraphics[width=0.9\linewidth]{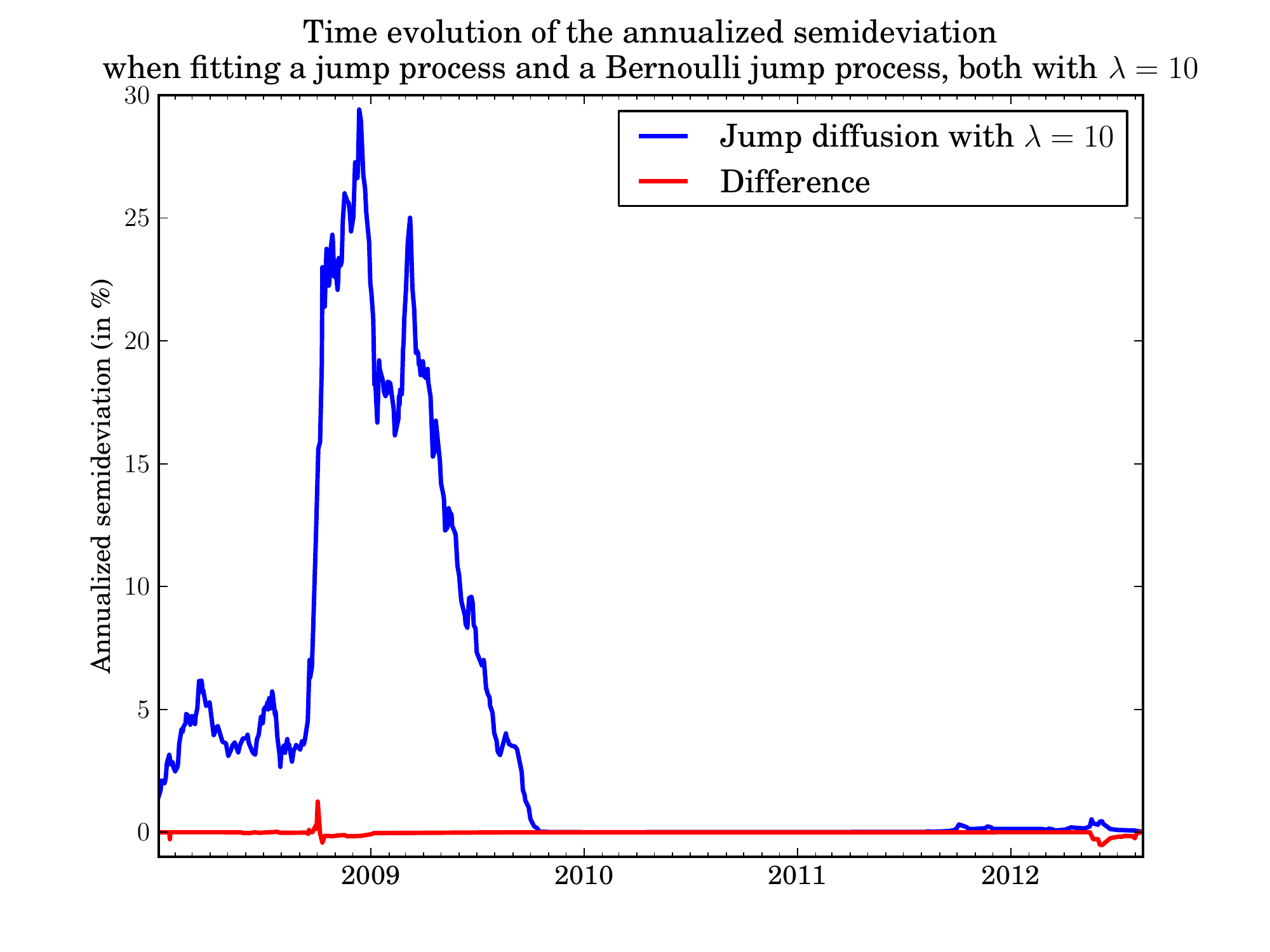}
      \caption{\label{graph:sds-jump-bern}Comparison between the
        semideviation (in \%) obtained from a jump
        diffusion process (in blue) and the difference with a Ball and Torous
        jump-diffusion (in red). In both cases, we have set $\lambda =
        10$. The difference is very small compared to the level of the
        semideviation.}
    \end{center}
  \end{figure}
\end{center}

\subsubsection{Evolution of $\sigma$}
Finally, we compare in Figure~\ref{graph:sigmas} the time
evolution of the volatility $\sigma$ when fitting a jump process and a
pure diffusion process. One can observe that the $\sigma$ obtained for a
pure diffusion process is larger in magnitude than the one obtained
for a jump diffusion process; this comes without any surprise, as in
the second case, the volatility is also captured by the jumps. A
second fact that we would like to outline is that $\sigma$ appears to
be quite unstable for the jump process. This can be interpreted as a
possible misspecification of the model, as the jumps
appear to influence the value of $\sigma$. In other words, assuming a
constant $\sigma$ might be a wrong hypothesis for our data, and this
can be translated as a further motivation to study processes with stochastic volatility including 
jumps in returns and in volatility, as we will do it in the next section. 
\begin{center}
  \begin{figure}
    \begin{center}
      \includegraphics[width=0.9\linewidth]{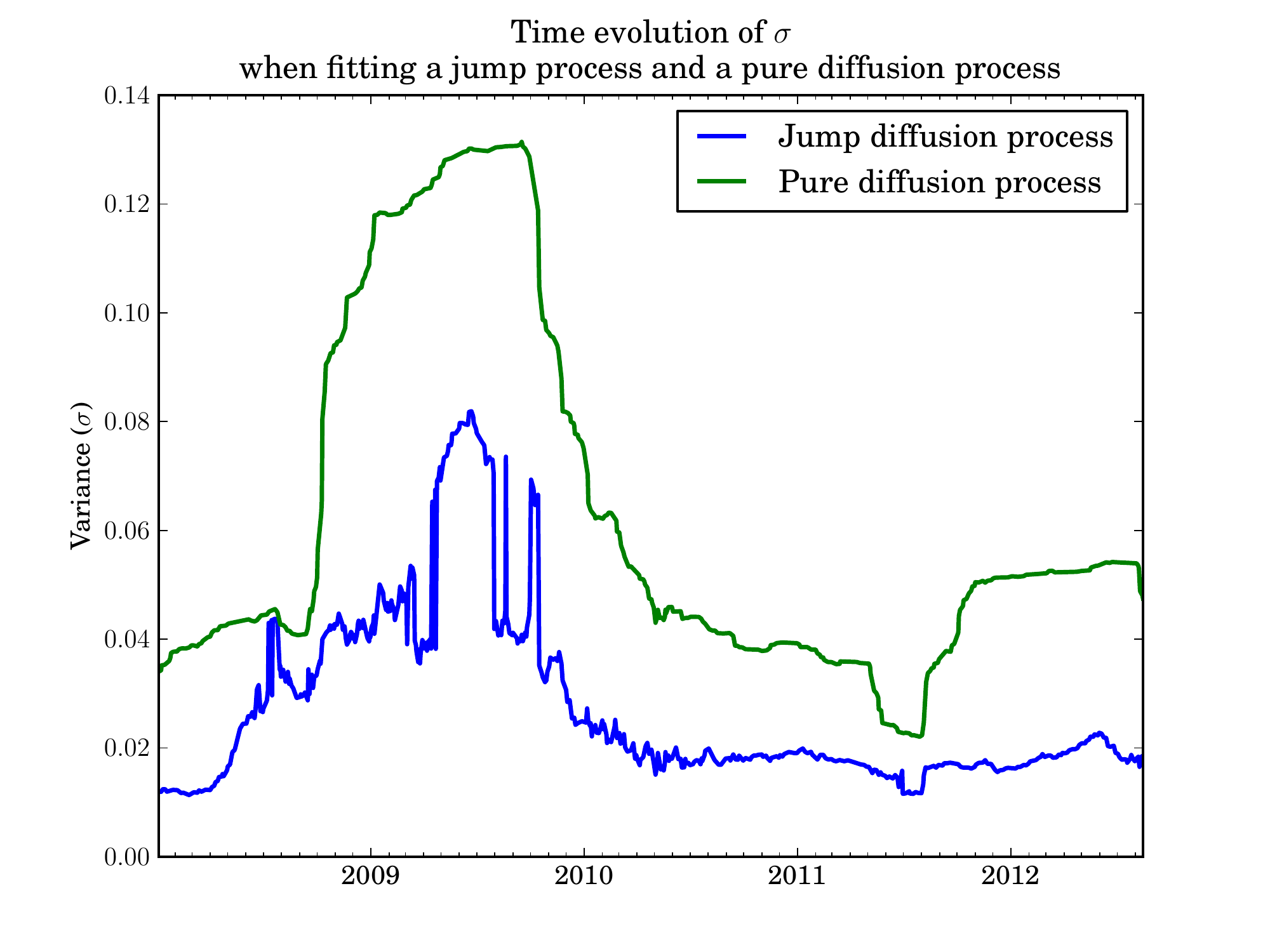}
      \caption{\label{graph:sigmas}Comparison between the time
        evolution of the parameter $\sigma$ obtained when fitting a
        jump process (in blue) and a pure
        diffusion process (in green).}
    \end{center}
  \end{figure}
\end{center}

\section{Extension to Stochastic Volatility Models with Jumps in
  Returns and Volatility}
\label{sec:jump_vol}
 
We explain now how to extend the methodology we developed before to
processes involving stochastic volatility with jumps in returns and
volatility. This section is structured as
follows. In~\S\ref{sub:import}, we recall
the importance of considering jumps in returns and in
volatility. Then, in~\S\ref{subsec:eraker}, we describe the stochastic
volatility model with jumps in returns and in volatility that we
consider, which is an extension of the model analyzed by
~\cite{EraJohPol03}. The only difference is
that we have replaced the Bernoulli jump process by the more general
model that we have developed in~\S\ref{sec:generalization}. We
continue by addressing in~\S\ref{sub:prac} the statistical estimation
of the model parameters and of their posterior distributions. We have
followed the approach developed in \cite{NumRen10} and our
implementation is based on a script developed in the R language
provided by the same authors that we have modified 
in order to be able to cope with our approach. Once the
parameters have been estimated, we show in~\S\ref{sub:semi} how to
compute an annualized semivariance. In a nutshell, the first step
consists in generating the returns over a 1-year horizon and then to
compute the semivariance thanks to numerical integration. Finally, we
exhibit some experimental results in~\S\ref{stochvol:res}.

\subsection{Importance of Jumps in Returns and Volatility}
\label{sub:import}

Modeling accurately the price of stocks is an important topic in
financial engineering. A major breakthrough was the~\cite{BlaSch73}
model, that however suffers from at least two major
drawbacks. The first one is that the stock price follows a lognormal
distribution and the second one is that its volatility is
constant. Several studies (see~\cite{CheGalGhyTau99}, for instance)
have emphasized that the asset returns' unconditional distributions
show a greater degree of kurtosis than implied by the normality
assumption, and that volatility clustering is present in the data,
suggesting random changes in returns' volatility. An extension of this
model is to integrate large movements in the stock prices making
possible to model financial crashes, like the Black Monday (October
19th, 1987). They were introduced in the form of jump-diffusion
models, see \cite{Mer76,CoxRos76}. An important extension of the
Black-Scholes model was to use stochastic volatility rather than
considering it as constant, see for instance
\cite{HulWhi87,Sco87,Hes93}. \cite{Bat96} and~\cite{Sco97} combined
these two approaches
and introduced the stochastic volatility model with jumps in
returns. Event though their new approach has helped better characterize 
the behavior of stock prices, several studies have shown that models
with both diffusive stochastic volatility and jumps in return were
incapable of fully capturing the empirical features of equity index
returns, see for instance~\cite{BakCaoChe97,Bat00,Pan02}. Their main
weakness is that they do not capture well the empirical fact that the
conditional volatility of returns rapidly increases. By adding jumps
in the volatility, then the volatility process is better
specified. \cite{DufPanSin00} were the first
to introduce a model with jumps in both returns and
volatility. \cite{EraJohPol03} have shown
that the new model with jumps in volatility performed better than
previous ones, and resulted in no major misspecification in the
volatility process.

\subsection{Extension of the work of Eraker, Johannes, and Polson}
\label{subsec:eraker}

As mentionned above, \cite{EraJohPol03} consider a jump
diffusion process with jumps in returns and in volatility. These jumps
arrive simultaneously, with the jump sizes being correlated. According
to the model, the logarithm of an asset's price $Y_t = \log(S_t )$
solves
\begin{equation}
\left( \begin{array}{c} dY_t \\ dV_t \end{array} \right) =
\left( \begin{array}{c} \mu \\ \kappa (\theta-V_t^-) \end{array}
\right) dt 
+ \sqrt{V_{t^-}} \left( \begin{array}{cc} 1 & 0 \\ \rho \sigma_\nu &
    \sqrt{(1-\rho^2)} \sigma_\nu \end{array} \right) d\bm{W}_t 
+ \left( \begin{array}{c} \xi^y\\ \xi^\nu \end{array}
\right) dN_t
\label{eqn:stochVolDef}
\end{equation}
where $V_t^− = \lim_{s \uparrow t} V_s$, $\bm{W}_t = ( W_t^{(1)}  \
W_t^{(2)})^\mathsf{T}$ is a standard Brownian motion in $\mathbb{R}^2$ and
$(\cdot)^\mathsf{T}$ denotes the transpose of a matrix or a vector. The
jump arrival $N_t$ is a Poisson process with
constant intensity $\lambda$, and this
model assumes that the jump arrivals for the returns and the
volatility are contemporaneous. The variables $\xi^y$ and $\xi^\nu$
denote the jump sizes in returns and volatility, respectively. The
jump size in volatility follows an exponential distribution $\xi^\nu
\sim  \exp(\mu_v)$ while the jump sizes in returns and volatility are
correlated with $\xi^y|\xi^\nu \sim  \mathcal{N}(\mu_y+\rho_j
\xi^\nu,\sigma_y^2)$. 

Their methodology relies on Markov Chain Monte Carlo (MCMC)
methods. The basis for their MCMC estimation is the
time-discretization of~\eqref{eqn:stochVolDef}:
\begin{eqnarray}
Y_{t + \Delta t} - Y_t & = &   \mu \Delta t  + \sqrt{V_{t } \Delta t}
\varepsilon_{t + \Delta t}^y + \xi_{t + \Delta t}^y J_{t + \Delta t},
\nonumber \\
V_{t + \Delta t} - V_t & = &   \kappa(\theta-V_{t}) \Delta t +
\sigma_\nu \sqrt{V_{t}  \Delta t}  \varepsilon_{t + \Delta t}^\nu   +
\xi_{t + \Delta t}^\nu J_{t + \Delta t}, \label{eqn:disc}
\end{eqnarray}
where $J_{t + \Delta t}^k = 1$ indicates a jump
arrival. $\varepsilon_{t  + \Delta t}^y$, $\varepsilon_{t + \Delta
  t}^\nu$ are standard normal random variables with correlation $\rho$
and $\Delta t$ is the time-discretization interval (i.e., one
day). The jump sizes retain their distributional structure and the
jump times are Bernoulli random variables with constant intensity
$\lambda \Delta t$. The authors apply
then Bayesian techniques to compute the model parameters
$\bm{\Theta}$. The posterior distribution summarizes the sample
information regarding the parameters $\bm{\Theta}$ as well as the
latent volatility, jump times, and jump sizes:
\begin{equation}
\Pr(\bm{\Theta},J, \xi^y,\xi^\nu,V|\bm{Y}) \propto
\Pr(\bm{Y}|\bm{\Theta},J,\xi^y,\xi^\nu,V)
\Pr(\bm{\Theta},J,\xi^y,\xi^\nu,V) 
\end{equation}
where $J,\xi^y,\xi^\nu,V$ are vectors containing the time series of
the relevant variables. The posterior combines the likelihood
$\Pr(\bm{Y}|\bm{\Theta},J,\xi^y,\xi^\nu,V)$ and
the prior
$\Pr(\bm{\Theta},J,\xi^y,\xi^\nu,V|\bm{Y})$. As
the posterior distribution is not known in closed form, the MCMC-based
algorithm generates samples by iteratively drawing from the following
conditional posteriors:
\begin{eqnarray*}
\text{Parameters:} &  \Pr(\bm{\Theta}_i | \bm{\Theta}_{-i}, J,
\xi^y, \xi^\nu, V, \bm{Y}) &
i = 1,\dots, k \\
\text{Jump times:} &  \Pr(J_t = 1 | \bm{\Theta}, \xi^y, \xi^\nu, V, \bm{Y}) & i
= 1,\dots, T \\
\text{Jump sizes:}&  \Pr(\xi_t^y | \bm{\Theta}, J_t = 1, \xi^\nu, V, \bm{Y}) & i
= 1,\dots, T \\
 &  \Pr(\xi_t^\nu | \bm{\Theta}, J_t = 1, V, \bm{Y}) & i = 1,\dots, T \\
\text{Volatility:}&  \Pr(V_t|  V_{t+\Delta t}, V_{t-\Delta t},\bm{\Theta}, J,
\xi^y, \xi^\nu, \bm{Y}) & i = 1,\dots, T
\end{eqnarray*}
where $\bm{\Theta}_{-i}$ denotes the elements of the parameter vector,
\emph{except} $\bm{\Theta}_i$. Drawing from these distributions is
straightforward, with the exception of the volatility, as its distribution
is not of standard form. To sample from it, the authors use a
random-walk Metropolis algorithm. For $\rho$, they use an independence
Metropolis algorithm with a proposal density centered
at the sample correlation between the Brownian increments. For a
review of these MCMC techniques, we recommend \cite{JohPol09} 
where the authors provide a review of the theory behind MCMC
algorithms in the context of continuous-time financial econometrics. 
The algorithm of~\cite{EraJohPol03} produces
a set of draws $\left\{\bm{\Theta}^{(g)}, J^{(g)}, \xi^{y(g)},
  \xi^{\nu(g)}, V^{(g)}\right\}_{g=1}^G$ which are samples from
$\Pr(\bm{\Theta},J,\xi^y,\xi^\nu,V|\bm{Y})$.

Our approach relies on the methodology developed
in~\cite{EraJohPol03}. The only difference is that we have replaced
the Bernoulli jump process by the jump process described
in~\S\ref{sec:generalization}.
This change is quite straightforward. In~\eqref{eqn:disc},  $J_{t +
  \Delta t}$ can take only two values and its posterior is Bernoulli
with parameter $\lambda\Delta t$. To compute the Bernoulli
probability, the authors use the conditional independence of
increments to volatility and returns to get 
\begin{eqnarray*}
&\Pr(J_{t+\Delta t} = 1 | V_{t + \Delta t}, V_t, Y_{t + \Delta t},
\xi_{t+\Delta t}^y, \xi_{t+\Delta t}^\nu, \bm{\Theta}) \propto & \\
& \lambda \Delta t \times \Pr(Y_{t + \Delta t}, V_{t + \Delta t} | V_t,
J_{t+\Delta t} = 1,  \xi_{t+\Delta t}^y , \xi_{t+\Delta t}^\nu,
\bm{\Theta}).& 
\end{eqnarray*}
Our generalization is straightforward. Following the approach
explained in~\S\ref{sec:generalization}, $J_{t + \Delta t}$ can take
any value between 0 and $m$. Then
 \begin{eqnarray*}
&\Pr(J_{t+\Delta t} = m | V_{t + \Delta t}, V_t, Y_{t + \Delta t},
\xi_{t+\Delta t}^y, \xi_{t+\Delta t}^\nu, \bm{\Theta}) \propto & \\
& p_k \times \Pr(Y_{t + \Delta t}, V_{t + \Delta t} | V_t, J_{t+\Delta
  t} = m,  \xi_{t+\Delta t}^y , \xi_{t+\Delta t}^\nu, \bm{\Theta}),&  
\end{eqnarray*}
with $p_k = \frac{e^{-\lambda \Delta t} (\lambda \Delta t)^k}{k!}$ for $k =
0,\dots, m-1$ and with $p_m = \sum_{k=0}^{m-1}p_k$.

\subsection{Parameters Estimation}
\label{sub:prac}

We have followed the approach developed by~\cite{NumRen10} for the
estimation of the model parameters. These
authors describe an implementation in R of the
methodology exposed by~\cite{EraJohPol03} and apply it on FTSE~100
daily returns. We assume that we have daily data at times $t_i$ for $i
= 1,\dots,T$ and
with $t_{i+1}-t_i = \Delta t = 1\text{ day}$. The bivariate density
function under consideration is given by
\begin{equation}
f(\bm{B}_{t_i}) = \frac{1}{2 \pi | \bm{\Sigma} | ^{ 0.5}} \exp \left(
  -\frac{1}{2}(\bm{B}_{t_i}-\mathrm{E}[\bm{B}_{t_i}])^\mathsf{T} \bm{\Sigma}_{t_i}^{-1}
  (\bm{B}_{t_i}-\mathrm{E}[\bm{B}_{t_i}]) \right),
\label{eqn:bi}
\end{equation}
where $|\cdot|$ denotes the determinant of matrix. The likelihood function is
simply given by $\prod_{i=1}^n  f(\bm{B}_{t_i})$ with
\begin{equation}
\bm{B}_{t_i}  =  \left( \begin{array}{c} \Delta y_{t_i} \\ \Delta
    \nu_{t_i} \end{array} \right)
\end{equation}
\begin{equation}
\mathrm{E}[\bm{B}_{t_i}]  =  \left(   \begin{array}{c} \mu + \xi_{t_i}^y \xi
    J_{t_i} \\ \kappa(\theta - V_{t_{i-1}}) + \xi_{t_i}^\nu
    J_{t_i} \end{array}  \right) \\
\end{equation}
\begin{equation}
\bm{\Sigma}_{t_i} = \mathrm{Cov}[\Delta y_{t_i}, \Delta \nu_{t_i}]   =
\left(   \begin{array}{cc} V_{t_{i-1}} &  \rho \sigma_\nu V_{t_{i-1}}
    \\ 
\rho \sigma_\nu V_{t_{i-1}} & \sigma_\nu^2 V_{t_{i-1}}
 \end{array}  \right) \\
\end{equation}
where $\Delta y_{t_i} = Y_{t_i} - Y_{t_{i-1}}$ and $\Delta
\nu_{t_i}=V_{t_i} - V_{t_{i-1}}$. The joint distribution is given by
the product of the likelihood times the distributions of the state
variables times the priors of the parameters, more specifically:
\begin{eqnarray*}
&&\left[ \prod_{i=1}^T f(\bm{B}_{t_i}) \right] \times \left[
  \prod_{i=1}^T
  f(\xi_{t_i}^y) \times f(\xi_{t_i}^\nu)\times f(J_{t_i}) \right] \\
&& \times \left[ f(\mu) \times f(\kappa) \times f(\theta) \times
  f(\rho) \times f(\sigma_\nu^2) \times f(\mu_y) \times f(\rho_J)
  \times f(\sigma_y^2)  \right] \\
&&  \times \left[ f(\mu_{\nu}) \times f(\lambda) \right]
\end{eqnarray*}
The distributions of the state variables are respectively given by
$\xi_{t_i}^\nu\sim \exp(\mu_\nu)$, $\xi_{t_i}^y \sim \mathcal{N}(\mu_y
+ \rho_J \xi_{t_i}^\nu$, $\sigma_y^2)$, and $J_{t_i} \sim
\mathcal{B}(\lambda)$. In~\cite{NumRen10}, the authors impose
little information through priors, that follow the same distributions
than in~\cite{EraJohPol03}: $\mu \sim \mathcal{N}(0,1)$, $\kappa \sim
\mathcal{N}(0,1)$, $\theta \sim \mathcal{N}(0,1)$, $\rho \sim
\mathcal{U}(-1,1)$, $\sigma_\nu^2 \sim\mathcal{IG}(2.5, 0.1)$, $\mu_y
\sim \mathcal{N}(0,100)$, $\rho_J \sim \mathcal{N}(0,1)$, $\sigma_y^2 \sim
\mathcal{IG}(5,20)$, $\mu_\nu\sim \mathcal{G}(20,10)$, and $\lambda \sim
\beta(2,40)$.\footnote{Here,
  $\mathcal{B}(\cdot)$, $\mathcal{G}(\cdot, \cdot)$,
  $\mathcal{IG}(\cdot, \cdot)$,
  $\mathcal{U}(\cdot, \cdot)$, $\beta(\cdot, \cdot)$ denote the
  Bernoulli, gamma, inverse gamma, uniform and beta distributions,
  respectively.} Following our
approach, we only made the following change: we have relaxed the
assumption about the Bernoulli process to consider the more general
process proposed in~\S\ref{sec:generalization}. It may be
also tempting to change the assumptions about the distribution of
$\lambda$, but the $\beta$ distribution is
so flexible that it can cope with a wide variety of assumptions about
these parameters. For example, the $\beta(1,1)$ gives a uniform
variable which can be used if we want to have a uninformative
distribution. 

It is important to store the sampled parameters and the vector of
state variables at each iteration. The mean of each parameter over the
number of iterations gives us the parameter estimate. The convergence
is checked using trace plots showing the history of the chain for each
parameter. ACF plots is used to analyze the correlation structure of
draws. Finally, the quality of the fit may be assessed through the
mean squared errors and the normal QQ plot. The standardized error is
defined as follows: 
$$ \frac{Y_{t+\Delta t}-Y_t - \mu \Delta t - \xi_{t+\Delta t}^y
  J_{t+\Delta t}^y}{\sqrt{V_{t+\Delta t} \Delta t}}=
\varepsilon_{t+\Delta t}^y.$$

\subsection{Computing the Semideviation}
\label{sub:semi}

Once we have managed to estimate the parameters of the model
on daily data, we still have to compute an annualized
semideviation. The first step consists in simulating the returns over
the $\tau$ horizon, where $\tau$ is typically 1-year if we are
interested in computing an annualized semideviation. We start by
simulating a Poisson process $N_{\tau}$ with jump intensity
$\lambda$. The output of this simulation consists of the number $N$ of
jumps occurring between times 0 and $\tau$, in addition to the times $0
\leq j_1 \leq \dots \leq j_N \leq \tau$ at which these jumps occur. For
each interval $[j_{i-1}, j_i]$, we simulate the diffusion parts of the
return and volatility processes according to an Euler discretization
schema for the Heston model, this procedure being explained later in
the section. This gives us preliminary values $Y_{j_i}^-$ and
$V_{j_i}^-$ for the return and variance processes at time $j_i$. Next,
we simulate jump sizes for the jumps in the two processes. These jumps
are given by $\xi_{j_i}^{y}$ and $\xi_{j_i}^{\nu}$, respectively. The
final values of the two processes at time $j_i$ can be calculated as:
\begin{eqnarray}
Y_{j_i} & = & Y_{j_i}^- + \xi_{j_i}^{y} \\
V_{j_i} & = & V_{j_i}^- +   \xi_{j_i}^{\nu}
\end{eqnarray}
If $j_N \neq \tau$, then no jump occurs in the interval $[j_N, \tau]$,
and we can apply the Euler discretization schema for the Heston model to
get the values of $Y_{\tau}$ and $V_{\tau}$. 

In the following, we explain how to implement the Euler discretization
schema for the Heston model. Discretizing the two
equations~\eqref{eqn:stochVolDef} and ignoring the jump part, the
difference between $Y_{t+\Delta t}$ and $Y_t$ and the difference
between $V_{t+\Delta t}$ and $V_t$ is simply given by 
\begin{eqnarray}
Y_{t+\Delta t} - Y_t & = &   \mu \Delta t  + \sqrt{V_t}  \left( \Delta
W_t^{(1)}\right), \nonumber \\
V_{t+\Delta t} - V_t & = &   \kappa(\theta-V_{t}) \Delta t +
\sigma_\nu \sqrt{V_{t}}  \left(\rho \Delta W_{t}^{(1)} + \sqrt{1-\rho^2}
\Delta W_{t}^{(2)}\right). \label{eqn:Heston}
\end{eqnarray}
As already mentioned before, we typically use $\Delta t=1$ day to
discretize both processes. To simulate the Brownian increments $\Delta
W_{t}^{(1)}$ and $\Delta W_{t}^{(2)}$, we use the fact that each
increment is independent of the others. Each such increment is normally
distributed with mean 0 and variance $\Delta t$. However,
using~\eqref{eqn:Heston} for the simulations may generate negative
volatilities. This is a well-known effect that has been addressed 
several times in the literature. Following~\cite{LorKoeDij06}, we
slightly modify~\eqref{eqn:Heston} to prevent the simulation
from generating negative values:
\begin{equation}
V_{t+\Delta t} - V_t  =   \kappa(\theta-V_{t}^+) \Delta t + \sigma_\nu
\sqrt{V_{t}^+}  \left( \rho \Delta W_{t}^{(1)} + \sqrt{1-\rho^2} \Delta
W_{t}^{(2)}\right),\label{eqn:Heston2}
\end{equation}
where $V_t^+$ is the maximum between $0$ and $V_t$. If the time
between
different jumps is smaller that the discretization step $\Delta t$,
then it simply means that we have several jumps occurring during that
time interval. In that case, we just have to simulate the
appropriate number of jumps, both for the return and for the
volatility, during that period. In the end, the return over the period
$[0, \tau]$ is just $Y_{\tau}$. By repeating this procedure $M$ times,
we get $M$ realizations of the accumulated returns over the period
$[0,\tau]$, which make it possible to estimate its density. In our
tests, we have used a Gaussian kernel to estimate it. This can easily
be done in R by using the \texttt{density()} function. Then, one can
obtain the semivariance by numerically integrating the expression
given in~\eqref{eq:TSV}. For this, we have used the QUADPACK
library via the use of the \texttt{integrate()} function in R. QUADPACK is a
well known numerical analysis library implemented in FORTRAN~77 and
containing very efficient methods for numerical integration of
one-dimensional functions.  

\subsection{Results}
\label{stochvol:res}

We ran the algorithm using 100'000 iterations with 10'000 burn-in
iterations. Table~\ref{tbl:param} provides parameter posterior means and their computed standard
errors. 
\begin{table}
  \begin{center}
    \begin{tabular}{lcc}
      \hline
      Parameter & Value (100k iter.) \\
      \hline
      $\mu$& 0.0333 (0.0132)\\
      $\mu_y$& -0.0222 (0.2055)\\
      $\mu_\nu$& 0.0028 (0.0007) \\
      $\theta$& 0.2087 (0.0147)\\
      $\kappa$& 0.9059 (0.0433)\\
      $\rho$&  -0.0058 (0.2595)\\
      $\rho_J$& 0.0050 (2.0007)\\
      $\lambda$&  0.0236 (0.0293)\\
      $\sigma_y$& 1.0829 (0.1379)\\
      $\sigma_\nu$& 0.2060 (0.0204)\\
      \hline
    \end{tabular}
  \end{center}
  \caption{Estimated parameters based on 100'000 iterations. The time
    unit is the day rather than a year, meaning that parameters should
    be interpreted on a daily basis. $\mu$, $\mu_y$, $\mu_\nu$,
    $\sigma_y$, and $\sigma_\nu$ are in \%.}
  \label{tbl:param}
\end{table}
The return mean $\mu$ is close to the daily return mean from the
data ($0.0325 \%$). The long-term mean of the $V_t$, given by
$\mathrm{E}[V_t] = \theta+(\mu_\nu \lambda) / \kappa$, is equal to
$0.2088 \%$. This value is not far away from the variance of returns,
which is equal to $0.1751$. The jump returns are normally distributed
with mean $-0.0222 \%$
and standard deviation $0.2055 \%$. Concretely, it means that there is a
$68 \%$ likelihood of having a jump return between $-23\text{ bps}$
and $+18\text{ bps}$. The parameters $\sigma_y$ and $\sigma_\nu$ are
$1.0829 \%$ and $0.2060 \%$, respectively. Finally, the intensity of the
jumps is $0.0236$. This means that, on average, there are
approximately 6 jumps per year. This is a bit low in comparison with
the number of jumps that we have observed in our data, as we would rather
expect an intensity around $0.05$. However, this empirical intensity
was computed considering that differences in returns above $2.57$
deviations from the mean could be considered as jumps. Some caution
must be taken, since this empirical intensity is very sensitive to the
number of standard deviations used as threshold for defining the
jumps. More details about the implementation can be found in
Appendix~\ref{sec:input}.

In a second step, we have computed an annualized semideviation based
on that process. We have followed the methodology explained
in~\S\ref{sub:semi} to first compute the simulated returns and
then the semideviation. 
\begin{center}
  \begin{figure}
    \begin{center}
      \includegraphics[width=0.9\linewidth]{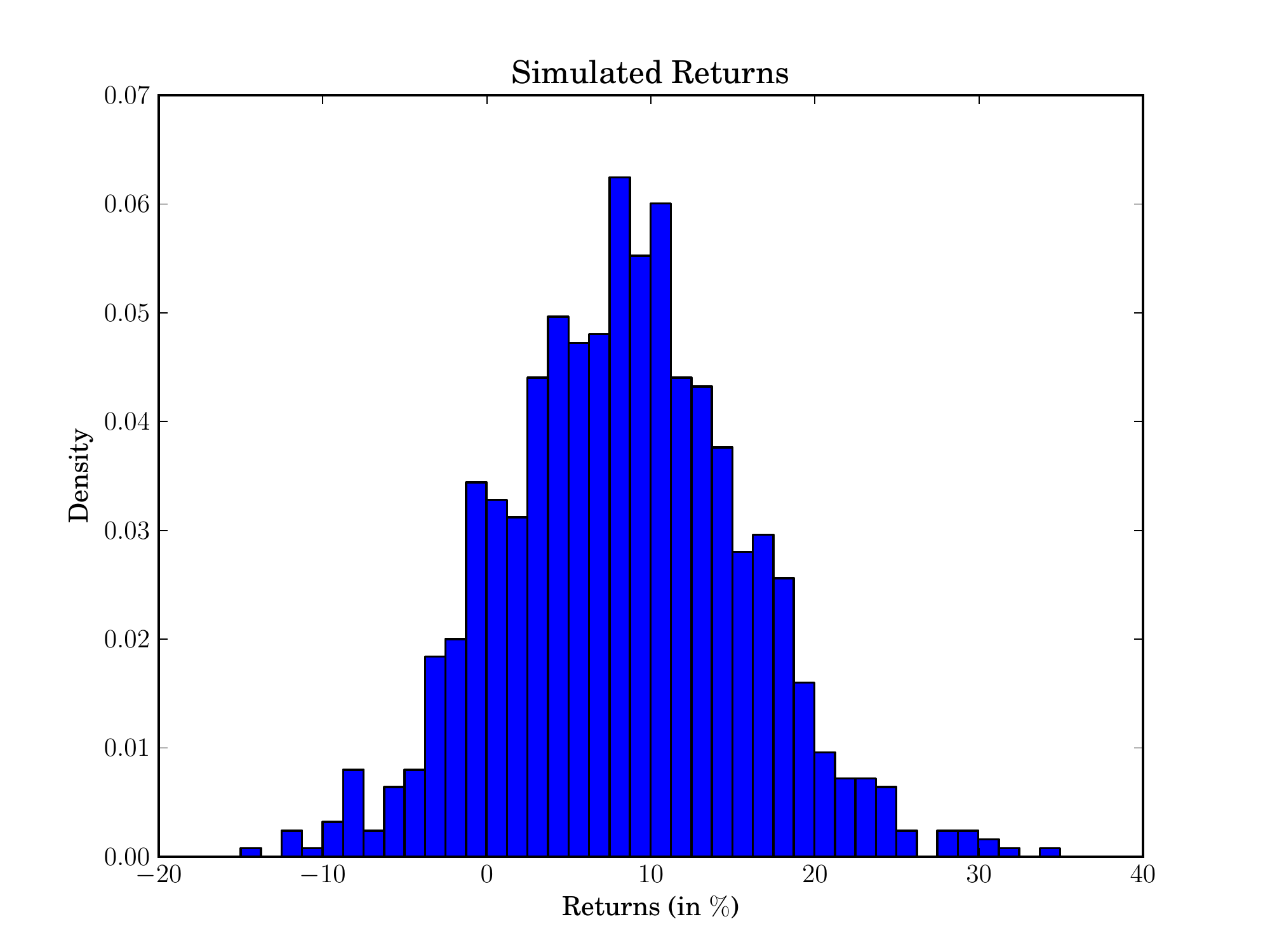}
      \caption{\label{graph:simulatedReturns}This graph shows the
        histogram of the returns based on 1'000 simulations.}
    \end{center}
  \end{figure}
\end{center}
Our tests have shown that 1'000 simulations were sufficient to get an
accurate estimation of the density from which we have derived a
semideviation. Note that this semideviation is based on the whole
historical data. It was not possible for us to do the same rolling
analysis that we have performed before, assuming a constant
volatility, for the two following reasons. The first one is that the
estimation of the parameters is very time-consuming. Around 7 hours
were necessary to generate the parameters for the whole period with an
\textit{Intel\/$^{\mbox{\scriptsize{\textregistered}}}$
  Core$^{\mbox{\scriptsize{\texttrademark}}}$~i5-3470 CPU @ 3.20
  GHz}. Assuming that we would have to perform a rolling analysis
based on 1-year of history over the same period, then we would have to
run more than 1000 times this analysis. This is not impossible to
realize but would typically require parallel computing techniques. The
second reason for not having performed this analysis is that we have
to check the convergence of the process for each period we consider in
the rolling analysis, a task that is difficult to fully automate. As a
reminder, the convergence can be checked with trace plots showing the
history of the chain for each parameter, with ACF plots to analyze the
correlation structure of draws, and  the quality of the fit may be
assessed through the mean squared errors and the normal QQ plot. So
these two reasons have prevented us from doing this rolling analysis
when we consider the model with stochastic volatility and jumps in
returns and in volatility. However, we have performed some analyses
based on 1 year of data for the two worst periods: during the credit
crunch crisis in 2008 and in 2011. Table~\ref{tbl:stochvolres}
summarizes the results and compares the semideviations obtained from
the different models that we have explored in this paper.
\begin{table}
  \begin{center}
    \begin{tabular}{ccccccccccccccccc}
      \hline
      Year          & SD1 & SD2 & SD3 & SD4  \\
      \hline 
      2008-2012   &      4.70             & 1.21            &     1.39              &         1.74          \\
      2008           &     8.71              &  31.59         &      31.72           &         32.16          \\
      2011           &     3.69              &  1.31            &     1.34             &         2.90             \\
      \hline
    \end{tabular}
    \caption{\label{tbl:stochvolres}Annualized semideviations (in \%) for different periods. SD1
      is based on the square-root-of-time rule, SD2 on a jump-diffusion
      model with constant volatility, SD3 on a pure diffusion model with a constant volatility, and
      SD4 is based on a jump diffusion model with stochastic volatility
      and jumps in returns and in volatility.}
  \end{center}
\end{table}
We can observe that the semideviation
SD4 obtained with a stochastic volatility model seems to be more in line
with the semideviations SD2 (based on a jump-diffusion model) and SD3
(based on a pure diffusion model) rather than SD1 (square-root-of-time
rule). This is not very surprising since SD4 may be interpreted as a
refinement of SD2, explaining why the two analytics should not be so
different.

\section{Conclusion and Future Venues of Research}
In this paper, we propose a generalization of the Ball-Torous
approximation of a jump-diffusion process and we show how to compute
the semideviation based on
several jump-diffusion models. We have in particular considered
a very exhaustive example based on the Barclays US High Yield Index,
whose returns show negative skewness and fat tails. It is a common
practice to impose a bound on the Poisson process intensity parameter
to make sure that only the large accidents are captured by the jumps
process. However, our analysis clearly shows that the risk may be
underestimated in such a situation. Without constraining the intensity
parameter, the semideviation may be more than $50$ \% larger than the
one obtained by arbitrary limiting this parameter in order to capture
only large jumps. 

We see at least two reasons why our approach
should be preferred. The first one is that we do not impose any
arbitrary constraint on the intensity parameter. With our method, this
parameter is determined in order to obtain the best fit to the
data. Jumps may occur very often in certain periods and can be very
uncommon (as rare as a few ones per year) in others. The second reason
is that imposing an arbitrary limit on the intensity parameter may
result in an underestimation of the risk. Capturing the risk as
accurately as possible is one of the most important mission in risk
management. It is the reason why the ``traditional" approach based on
the work of Ball and Torous does not seem appropriate in our
context. Moreover, one of our conclusions is that the use of the
square-root-of-time rule may either underestimate the risk in periods
of high volatility or underestimate it in periods of lower stress.

We also provide in this paper a generalization of the work of 
Eraker, Johannes and Polson, who consider a jump diffusion model with
stochastic volatility and jumps in returns and in volatility. As in
the case where the volatility is constant, we have extended their
approach by replacing the jump Bernoulli process by a more general
process being able to model several accidents per day when the
intensity parameter is high. The parameter estimation methods are
based on MCMC methods. In particular, we have used a Gibbs sampler
when the conditional distributions were known and variants of the
Metropolis algorithm when it was not the case. We also provide a
procedure to compute an annualized semideviation once the parameters
of the model have been estimated.

It was not possible for us to go as far as we would like to when we
have considered the model with stochastic volatility with jumps in
returns and in volatility. It would have been very useful to be
able to do the same rolling analyses that we have performed when the
volatility was kept constant. However, such analyses would have been
too much time consuming from a computational point of view and it
would have been complicated to automate the check of the convergence
of the process that needs to be done for each period in the rolling
analysis. This was clearly a limitation when we have tried to
determine empirically the relationship between the daily semideviation
and its annualized version.   

We really think that the use of the semideviation (or semivariance)
could benefit the finance industry by
providing a useful and powerful risk measure. This risk metric has
been proposed a very long time ago but its difficulty to be computed
and its lack of nice properties for its scaling made it hard to
implement. However, we have shown in this article that it is still
possible to calculate it even when we consider very complex stochastic
processes and that some useful formula for its time scaling can be
derived under some mild assumptions. Finally, we hope that this paper will
help democratize its use in the asset management industry.\\
~\\
\noindent\textbf{Disclaimer:} The views expressed in this article are
the sole responsibility of the authors and do not necessarily reflect
those of Pictet Asset Management SA. Any remaining errors or
shortcomings are the authors’ responsibility.

\appendix

\section{Proof of Theorem~\ref{the:semivariance}}
\label{app:proof}

As a first step, we note the two following identities:
\begin{equation}
  \label{eq:easy_identities}
  \int_{-\infty}^{a} t\phi(t)\,dt = -\phi(a) \text{ and }
  \int_{-\infty}^{a} t^2\phi(t)\,dt = -a\phi(a) + \Phi(a). 
\end{equation}
The first one can be easily obtained by observing that
$\frac{d\phi}{dt} = -t\phi(t)$ and $\lim_{t\to-\infty} \phi(t) =
0$. Then, 
\begin{equation*}
  \int_{-\infty}^{a} t\phi(t)\,dt = -\int_{-\infty}^{a} \frac{d}{dt}
  \phi(t) \,dt = \phi(a). 
\end{equation*}
For obtaining the second identity, one can proceed as follows: we note
that, for a standardized normal density $\phi$, $\int t\phi(t)\,dt =
-\phi(t)$. Integrating by part, we get
\begin{equation*}
  \int_{-\infty}^{a}t^2\phi(t)\,dt =
  \int_{-\infty}^{a}t(t\phi(t))\,dt = -t\phi(t)\Bigl|_{-\infty}^{a} +
  \int_{-\infty}^{a} \phi(t)\,dt.
\end{equation*}
We can conclude by noting that $\lim_{t\to-\infty} t\phi(t) =
0$. Those two identities are useful to prove the following result. 
\begin{proposition}
  \label{prop:prop0}
Let $W \sim  \mathcal{N}(\mu_1,\sigma_1^2)$, with a density function  
$$f_W(w)=\frac{1}{\sigma_1 \sqrt{2  \pi}} \exp\left(-\frac{1}{2} \left(
    \frac{w-\mu_1}{\sigma_1}\right)^2\right),$$
and let $\tilde{D} = \frac{D-\mu_1}{\sigma_1}$; then, the semivariance
of $f_W$ is given by
\begin{equation}
\int_{-\infty}^D f_W(w)(D-w)^2dw  =   (D-\mu_1)^2 \Phi(\tilde{D}) +
\sigma_1 (D-\mu_1) f(\tilde{D}) +\sigma_1^2 \Phi(\tilde{D})
\label{eqn:semivarianceformula}
\end{equation}
\end{proposition}

\begin{proof}
By first using the change of variable $x = \frac{w-\mu_1}{\sigma_1}$
and the identities given in~\eqref{eq:easy_identities}, we obtain the
following straightforward development. 
\begin{eqnarray*}
\int_{-\infty}^D f_W(w)(D-w)^2dw  =  \int_{-\infty}^{\tilde{D}}
f(x)\left(D-(\sigma_1 x + \mu_1)\right)^2dx =  \\
\int_{-\infty}^{\tilde{D}} (D- \mu_1)^2 f(x) dx 
 +\int_{-\infty}^{\tilde{D}} 2 \sigma_1 (\mu_1-D) x f(x) dx  
 + \int_{-\infty}^{\tilde{D}} \sigma_1^2 x^2 f(x) dx = \\
(D-\mu_1)^2 \Phi(\tilde{D}) +2 \sigma_1 (D-\mu_1) f(\tilde{D}) 
 +\sigma_1^2 \left(-\tilde{D} f(\tilde{D}) + \Phi(\tilde{D})\right) = \\
 (D-\mu_1)^2 \Phi(\tilde{D}) + \sigma_1 (D-\mu_1) f(\tilde{D})
 +\sigma_1^2 \Phi(\tilde{D})
\end{eqnarray*}
\end{proof}
Taking into account that the log returns follow the density given
in~\eqref{eqn:dens0}, we can apply Proposition~\ref{prop:prop0}
on~\eqref{eqn:dens0} to obtain the formula of
Theorem~\ref{the:semivariance}.

\section{Use of a Jump-Diffusion Process to Model a Bond
  Benchmark Index}
\label{app:NAV}
Using a jump-diffusion process is a standard practice for modeling a
stock price. We explain here why it is possible to use this model for a bond benchmark index,
 because this may seem disturbing at a first sight. Indeed, the particularities of a stock and of a bond are
quite different. In particular, a bond has a well-defined maturity,
while it is not the case for a stock. Moreover, the clean price of a
bond converges to 100 at its maturity. This is the well-known bond
pull-to-par effect. 

However, when considering a bond benchmark, most of criticisms in
favor of not using an equity model in this context become
irrelevant. Indeed, bond benchmarks have some characteristics that are
quite different from a bond. For example, they are typically
rebalanced on a monthly basis and their maturity are (almost) kept
constant, while the maturity of a bond decreases linearly with time.
The bonds with the shortest maturity (for example less than 1 year)
are automatically removed from the index.\footnote{The technical
  document describing the management of the Barclays bond indexes can
  be obtained at \url{https://indices.barcap.com/index.dxml}.} A
direct consequence of holding the maturity (almost) constant is that
the pull-to-par effect at the portfolio level is not relevant anymore.

Another distinction between a bond and a stock is that a bond pays
coupons while a stock pays dividends. When a coupon is paid, its dirty price
falls for an amount equivalent to the coupons. However, it is not
necessary to model the coupon payments for bond benchmarks, since the
rules they obey imply that the coupons are automatically
reinvested. Concretely, it means that the payment of a coupon does not
impact the benchmark index. 

Moreover,  in a diffusion model, the stock price
increases exponentially if we ignore the random part of the
model. This is a direct consequence of the fact the returns are not
additive, but must be compounded. This also applies to a bond
benchmark index where the performance is also
compounded. 

Finally, as stock markets may experience crashes, bond markets may
also suffered from extremely severe and sudden losses justifying the
use of jumps in returns and in volatility.  

All these remarks suggest that using a jump-diffusion process to model
a bond benchmark index seems to be appropriate, even though the
dynamics driving bonds are different from the ones driving stocks. But
at the portfolio level, and when we consider a bond benchmark with a
constant maturity, there is no fundamental reason for not
approximating the dynamics of the index level as if it were a stock
price. 

\section{Details about the Implementation}
\label{sec:input}
Our MCMC algorithm is implemented in R language and relies on the code
provided by~\cite{NumRen10}. We have modified
their implementation in order to cope with the model
that we have introduced in this paper. The starting values for our
algorithm were generated as follows. The volatility vector was created
using a three-month rolling window. We considered as jumps absolute
returns (the absolute value of the returns) that were above $2.57$
standard deviations above the mean (after accounting for
outliers). For a standardized normal distribution, $2.57$ corresponds
to its $0.995$-quantile, meaning that
we should expect that approximately $1\%$ of the absolute returns to
be larger that this threshold. Indeed, in our data, we had more than
$14\%$ of the data exceeding this threshold, confirming the presence
of jumps in returns. Once these ``accidents" have been identified, we
still need to determine the jump return size, the jump volatility size
and the number of jumps that have happened. Indeed, with our
methodology, it is
possible to have more than a single jump occurring during a specific date.
Provided that the jump sizes are by construction far larger than that
returns observed during the days without ``accidents", we
have considered that the jump size corresponds to the observed daily
returns. For example, if we have identified a jump and if the observed
return is -67 bps for that date, then we consider that the jump size
is the same amount. We have made a similar assumption for the
volatility. So the
volatility jump simply correspond to the difference in volatility
between the actual estimation for a particular day and the day
before. The process followed by the sum of the jump returns is a
Poisson compound variable with normally distributed jumps. The same
for the process followed by the sum of the volatility jumps, except
that the jumps are this time exponentially distributed. Provided that
the jumps can be negative or positive for returns but not for the
volatilities (always positive), it is much easier to consider
the volatility to determine the number of jumps. It is the reason why
we consider a compound Poisson process $X$ for the jumps in volatility
where the Poisson process results in a $\mathcal{P}(\lambda)$ and the jump
sizes are exponential variables $\xi^\nu$ of parameter
$\nu$. It is easy to show that the distribution function
$F(x)$ of the process $X$ is simply given by
$$ F(x) = \Pr(X \leq x) = \sum_{k=0}^{+\infty}p_k \Pr\left( \sum_{i
    =1}^k \xi^\nu \right),$$ 
where $p_k = e^{-\lambda}\frac{\lambda^k}{k!}$. If there is no jump,
then $X = 0$ with probability $p_0 = e^{-\lambda}$. A direct
consequence is that the density does not exist. In order to overcome this issue, 
we propose to estimate $\lambda$ based on the proportion of days having at least a jump. The
probability of having no accident in a particular day is
$e^{-\lambda} \approx 1- \lambda$, when $\lambda$ is small. From that
formula, we can simply estimate $\lambda$ by the ratio of days with at
least one jump over the total number of days within that
period. 


Once these parameters have been estimated, we still need to determine
the number of jumps occurring when an accident is detected. It is
well-known that the convolution of $k$ iid exponential variables is
given by a Gamma distribution with parameters $k$ and $\lambda$.
So, based on the observation of the jump size, we can determine the
number of jumps by determining which $k$ gives the highest density.\\

Having determined the state space variables, we can now tackle the
problem of estimating the model parameters. We have implemented either
the Gibbs sampler or Metropolis-Hasting depending on whether or not we
knew the closed form of the parameters' distributions. In order to able to
compare our results with those in~\cite{EraJohPol03} and
in~\cite{NumRen10}, we have used the same
priors. Note that little information is imposed through priors. They
are as follows: $\mu \sim \mathcal{N}(0,1), \kappa \sim
\mathcal{N}(0,1), \kappa\theta
\sim \mathcal{N}(0,1), \rho \sim \mathcal{U}(-1,1), \sigma_\nu^2 \sim
\mathcal{IG}(2.5, 0.1), \mu_y
\sim \mathcal{N}(0,100), \rho_J \sim \mathcal{N}(0,1), \sigma_y^2 \sim
\mathcal{IG}(5,20), \mu_\nu \sim \mathcal{G}(20,10)$, and $\lambda
\sim \beta(2,40)$. Another reason for us
to use parameters provided by other studies is to guarantee that the
priors have been determined independently from their posterior
distributions. In~\cite{EraJohPol03}, it was also shown that these
priors may lead to very different posteriors when applied to S\&P~500
and Nasdaq~100 index returns. Moreover, the main objective of our work
is not to challenge the estimation of the parameters and in particular the
assumptions about the priors, but merely to show that it is possible to
compute a semideviation even when we consider complex stochastic
processes. 




\end{document}